\documentclass[USletter,10pt]{article}
\usepackage[margin=1in,dvips]{geometry}
\usepackage{graphicx,cancel}
\usepackage{graphics}
\usepackage{color, soul}
\usepackage{amsmath, amsthm, slashed}
\usepackage{amssymb}
\usepackage{rotating}
\usepackage{mathrsfs}
\usepackage{epsfig}
\usepackage{verbatim}
\usepackage[affil-it]{authblk}
\usepackage{rotating}
\usepackage{graphicx}
\usepackage{accents}
\newtheorem{definition}{Definition}[section]
\newtheorem{theorem}{Theorem}[section]
 
\newtheorem*{conjecture*}{Conjecture}
\newtheorem{corollary}{Corollary}[section]
\newtheorem*{theorem*}{Theorem}
\newtheorem*{corollary*}{Corollary}
\newtheorem{proposition}{Proposition}[section]
\newtheorem{lemma}{Lemma}[section]
\newtheorem{remark}{Remark}[section]

\newcommand{\So}{\mathscr{S}}
\newcommand{\Sop}{\tilde{\mathscr{S}}}
\newcommand{\Sf}{\accentset{\land}{\mathscr{S}}}
\newcommand{\Si}{\underaccent{\lor}{\mathscr{S}}}

\newcommand{\Olin}{\Omega^{-1}\accentset{\scalebox{.6}{\mbox{\tiny (1)}}}{\Omega}}
\newcommand{\Olino}{\accentset{\scalebox{.6}{\mbox{\tiny (1)}}}{\Omega}}
\newcommand{\glinh}{\accentset{\scalebox{.6}{\mbox{\tiny (1)}}}{\hat{\slashed{g}}}}

  \newcommand{\glinto}{\accentset{\scalebox{.6}{\mbox{\tiny (1)}}}{\sqrt{\slashed{g}}}}
\newcommand{\bmlin}{\accentset{\scalebox{.6}{\mbox{\tiny (1)}}}{b}}

\newcommand{\xblin}{\accentset{\scalebox{.6}{\mbox{\tiny (1)}}}{\underline{\hat{\chi}}}}
\newcommand{\xlin}{\accentset{\scalebox{.6}{\mbox{\tiny (1)}}}{{\hat{\chi}}}}

\newcommand{\eblin}{\accentset{\scalebox{.6}{\mbox{\tiny (1)}}}{\underline{\eta}}}
\newcommand{\elin}{\accentset{\scalebox{.6}{\mbox{\tiny (1)}}}{{\eta}}}
\newcommand{\otx}{\accentset{\scalebox{.6}{\mbox{\tiny (1)}}}{\left(\Omega tr \chi\right)}}
\newcommand{\otxb}{\accentset{\scalebox{.6}{\mbox{\tiny (1)}}}{\left(\Omega tr \underline{\chi}\right)}}
\newcommand{\olin}{\accentset{\scalebox{.6}{\mbox{\tiny (1)}}}{\omega}}
\newcommand{\olinb}{\accentset{\scalebox{.6}{\mbox{\tiny (1)}}}{\underline{\omega}}}

\newcommand{\ablin}{\accentset{\scalebox{.6}{\mbox{\tiny (1)}}}{\underline{\alpha}}}
\newcommand{\alin}{\accentset{\scalebox{.6}{\mbox{\tiny (1)}}}{{\alpha}}}

\newcommand{\bblin}{\accentset{\scalebox{.6}{\mbox{\tiny (1)}}}{\underline{\beta}}}
\newcommand{\blin}{\accentset{\scalebox{.6}{\mbox{\tiny (1)}}}{{\beta}}}
\newcommand{\rlin}{\accentset{\scalebox{.6}{\mbox{\tiny (1)}}}{\rho}}
\newcommand{\slin}{\accentset{\scalebox{.6}{\mbox{\tiny (1)}}}{{\sigma}}}
\newcommand{\Klin}{\accentset{\scalebox{.6}{\mbox{\tiny (1)}}}{K}}

\newcommand{\Zlin}{\accentset{\scalebox{.6}{\mbox{\tiny (1)}}}{Z}}

\DeclareMathAlphabet\mathbfcal{OMS}{cmsy}{b}{n}

\title{Conservation laws and flux bounds  \\ for gravitational perturbations \\ of the Schwarzschild metric}

\author{Gustav Holzegel\thanks{g.holzegel@imperial.ac.uk}}

\affil{\small Imperial College London, Department of Mathematics, 
South~Kensington~Campus,~London~SW7~2AZ,~United~Kingdom\vskip.2pc \ }

\begin{document}
\maketitle

\abstract{We derive an energy conservation law for the system of gravitational perturbations on the Schwarzschild spacetime expressed in a double null gauge. The resulting identity involves only first derivatives of the metric perturbation. Exploiting the gauge invariance up to boundary terms of the fluxes that appear, we are able to establish positivity of the  flux on any outgoing null hypersurface to the future of the initial data. This allows us to bound the total energy flux through any such hypersurface, including the event horizon, in terms of initial data. We similarly bound the total energy radiated to null infinity. Our estimates provide a direct approach to a weak form of stability, thereby complementing the proof of the full linear stability of the Schwarzschild solution recently obtained in [M.~Dafermos, G.~Holzegel
and I.~Rodnianski
\emph{The linear stability of the Schwarzschild solution to gravitational perturbations}, arXiv:1601.06467].}

\section{Introduction}
The study of gravitational perturbations of the Schwarzschild spacetime \cite{schwarzschild1916} in general relativity originates in the pioneering work of Regge and Wheeler \cite{Regge} more than 50 years ago. Despite a deepened understanding as well as significant refinements and generalizations over the years, see for instance \cite{Chandrasekhar, Moncrief, martelpoisson, Whiting}, a proof of the full linear stability of the Schwarzschild solution was only given in a recent paper of the author in collaboration with Dafermos and Rodnianski \cite{DHR}. The main result of \cite{DHR} can be stated as follows:

\begin{theorem*}[\cite{DHR}]
General solutions $\Si$ of the system of gravitational perturbations on Schwarzschild arising from suitably normalised characteristic initial data remain uniformly bounded on the black hole exterior and in fact decay inverse polynomially to a linearised Kerr solution $\mathscr{K}$ after adding to $\Si$ a dynamically determined pure gauge solution $\mathscr{G}$, which is itself uniformly bounded by initial data.
\end{theorem*}

The proof of the above theorem relies on recent advances in the black hole stability problem  \cite{KayWald, Mihalisnotes, BlueSoffer, withYakov, Holzegelspin2}, specifically the complete understanding of the scalar wave equation on black hole exteriors. It provides a complete physical space theory for the study of gravitational perturbations adapting some of the classical insights into the structure of gravitational perturbations \cite{newmanpenrose, Chandrasekhar, bardeen1973, whitingsurvey}.  The key ingredients of the proof may be summarised as follows:
\begin{itemize}
\item expressing the linearised Einstein equations in a double null gauge
\item a complete understanding of ``pure gauge" solutions of the resulting linearised system (preserving the double null form) arising from the diffeomorphism invariance of the full non-linear theory
\item exploiting the existence of gauge invariant quantities which \emph{decouple} from the full system \cite{bardeen1973}, in particular the Teukolsky null curvature components $\alpha$ and $\underline{\alpha}$ 
\item a physical space transformation mapping solutions of the Teukolsky equation to solutions of the Regge-Wheeler equation, for whose solutions robust decay estimates can be derived (see \cite{Chandrasekhar} for a fixed frequency version of these transformations applied to individual modes)
\item identifying a hierarchical structure in the linearised system allowing to estimate all remaining quantities \emph{through transport equations}
\end{itemize}
As a byproduct of the proof in \cite{DHR} one obtains, from the estimates satisfied by the decoupled, gauge-invariant quantities alone, control over the gauge invariant flux\footnote{strictly speaking it is only a \emph{residual gauge invariance} as certain partial gauge choices on the horizon always need to be made in order to estimate the flux} of the linearised shear on the event horizon and the flux of the linearised (weighted) shear on null infinity. 
As the former can be interpreted as a measure of the total energy leaving through the horizon and the latter as a measure of total energy radiated to null infinity, one may view the boundedness of these fluxes as a weak form of stability. This boundedness alone can already be seen to imply, in particular, ``mode stability". 

While control over the aforementioned fluxes on the horizon and null infinity in terms of initial data is of course a much weaker stability statement than the full linear stability result of \cite{DHR} (the latter providing boundedness and decay for \emph{all} dynamical quantities), one may ask whether one can obtain the weaker stability statement directly \emph{without making use of the equations satisfied by decoupled quantities  such as $\alin$ and $\ablin$}. In this paper, we prove that this can indeed be done. We give here a  rough version of our main theorem, Theorem \ref{theo:mtheo} below:

\begin{theorem} \label{theo:mtintro}
General solutions $\Si$ of the system of gravitational perturbations arising from suitably normalised characteristic initial data satisfy the following properties: 
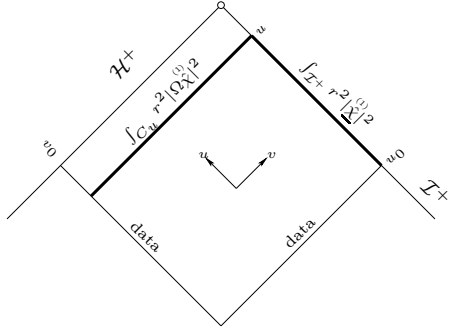
\[
\begin{picture}(0,0)%
\includegraphics{consintro3.pstex}%
\end{picture}%
\setlength{\unitlength}{1263sp}%
\begingroup\makeatletter\ifx\SetFigFont\undefined%
\gdef\SetFigFont#1#2#3#4#5{%
  \reset@font\fontsize{#1}{#2pt}%
  \fontfamily{#3}\fontseries{#4}\fontshape{#5}%
  \selectfont}%
\fi\endgroup%
\begin{picture}(8424,6385)(2689,-8173)
\put(5101,-6211){\rotatebox{315.0}{\makebox(0,0)[lb]{\smash{{\SetFigFont{6}{7.2}{\rmdefault}{\mddefault}{\updefault}{\color[rgb]{0,0,0}data}%
}}}}}
\put(8251,-6661){\rotatebox{45.0}{\makebox(0,0)[lb]{\smash{{\SetFigFont{6}{7.2}{\rmdefault}{\mddefault}{\updefault}{\color[rgb]{0,0,0}data}%
}}}}}
\put(5101,-4561){\rotatebox{45.0}{\makebox(0,0)[lb]{\smash{{\SetFigFont{6}{7.2}{\rmdefault}{\mddefault}{\updefault}{\color[rgb]{0,0,0}$\int_{C_u} r^2 |\Omega \xlin|^2$}%
}}}}}
\put(7651,-2461){\rotatebox{45.0}{\makebox(0,0)[lb]{\smash{{\SetFigFont{5}{6.0}{\rmdefault}{\mddefault}{\updefault}{\color[rgb]{0,0,0}$u$}%
}}}}}
\put(10201,-5011){\rotatebox{45.0}{\makebox(0,0)[lb]{\smash{{\SetFigFont{5}{6.0}{\rmdefault}{\mddefault}{\updefault}{\color[rgb]{0,0,0}$u_0$}%
}}}}}
\put(7801,-4861){\makebox(0,0)[lb]{\smash{{\SetFigFont{5}{6.0}{\rmdefault}{\mddefault}{\updefault}{\color[rgb]{0,0,0}$v$}%
}}}}
\put(6451,-4861){\makebox(0,0)[lb]{\smash{{\SetFigFont{5}{6.0}{\rmdefault}{\mddefault}{\updefault}{\color[rgb]{0,0,0}$u$}%
}}}}
\put(3376,-4561){\rotatebox{315.0}{\makebox(0,0)[lb]{\smash{{\SetFigFont{5}{6.0}{\rmdefault}{\mddefault}{\updefault}{\color[rgb]{0,0,0}$v_0$}%
}}}}}
\put(4876,-3361){\rotatebox{45.0}{\makebox(0,0)[lb]{\smash{{\SetFigFont{8}{9.6}{\rmdefault}{\mddefault}{\updefault}{\color[rgb]{0,0,0}$\mathcal{H}^+$}%
}}}}}
\put(8476,-3061){\rotatebox{315.0}{\makebox(0,0)[lb]{\smash{{\SetFigFont{6}{7.2}{\rmdefault}{\mddefault}{\updefault}{\color[rgb]{0,0,0}$\int_{\mathcal{I}^{\scalebox{.7}{+}}} r^2|\xblin|^2$}%
}}}}}
\put(10801,-5461){\rotatebox{315.0}{\makebox(0,0)[lb]{\smash{{\SetFigFont{8}{9.6}{\rmdefault}{\mddefault}{\updefault}{\color[rgb]{0,0,0}$\mathcal{I}^+$}%
}}}}}
\end{picture}%

\]

\begin{enumerate}
\item The total energy flux of the linearised shear $\xlin$ along {\bf any} outgoing null cone $C_{u}$ with $u \geq u_0$,
\[
\int_{C_u}  |\Omega \xlin|^2 r^2 dv \sin \theta d\theta d\phi \, ,
\] 
 is bounded by initial data, uniformly in $u$. 
\item The total energy flux along the event horizon $\mathcal{H}^+$, 
\[
\int_{\mathcal{H}^+} |\Omega \xlin|^2 r^2 dv \sin \theta d\theta d\phi  \, ,
\] 
is bounded by initial data. 
\item The total linearised gravitational energy flux radiated to null infinity $\mathcal{I}^+$
\[
\int_{\mathcal{I}^+} |\xblin|^2 r^2 du \sin \theta d\theta d\phi  \, ,
\]  
is bounded by initial data.
\end{enumerate}
\end{theorem}

We remark immediately that we have stated the theorem for characteristic initial data for convenience (and to make manifest the geometric significance of the terms appearing in the initial data). The asymptotic fluxes are similarly controlled for solutions arising from spacelike initial data.

As mentioned above, the stability statement implicit in Theorem \ref{theo:mtintro} is much weaker than the statement of \cite{DHR}. However, as we will see, the estimate itself requires less control on the initial data: While in~\cite{DHR} certain (up to) second derivatives of curvature were required to be bounded initially to obtain control of the linearised shear on the event horizon, we prove here that initial boundedness of some of the metric and connection coefficients is sufficient to control the total energy fluxes through the event horizon and null infinity. 

Key to the proof of Theorem \ref{theo:mtintro} are certain conservation laws inherent in the system of gravitational perturbations at the level of the linearised metric and connection coefficients.
The existence of such conservation laws for the system of gravitational perturbations on 
static and stationary axisymmetric vacuum and electro-vacuum spacetimes was pioneered by Friedman \cite{Friedman}, Chandrasekhar \cite{Chandrasekhar, Ferrari} and Wald and collaborators \cite{Lee, Burnett}, one motivation being a deeper conceptual understanding of the the reflection and transmission coefficients summing up to unity in the context of black hole scattering theory. At a very general level, one may understand the existence of conservation laws from the fact that the non-linear Einstein equations arise from a variational principle and that the background with respect to which the linearisation is performed admits a non-trivial timelike Killing field.  
Recently, the subject of conservation laws in black hole perturbation theory has been revived through important work of Hollands and Wald \cite{HollandsWald} (see also \cite{Iyer}), who introduced, for general static or stationary axisymmetric spacetimes, a notion of \emph{canonical energy} based on the symplectic structure on the space of perturbations.  See also \cite{Keir} for generalisations of the canonical energy.


For a general conservation law to be useful in controlling the dynamics of the system, one needs the associated energies to be coercive. Remarkably, the canonical energy of \cite{HollandsWald} in particular produces manifestly positive energy fluxes through the event horizon and null infinity for axisymmetric perturbations on stationary backgrounds, and for general perturbations if the background is static. Unfortunately, however, it has not yet been shown to produce non-negative energies on a foliation of spacelike slices connecting the event horizon and null infinity, even for gravitational perturbations on Schwarzschild, which is the case under consideration here. It is this fact which has prevented the direct use of the canonical energy to prove weak stability statements. 

On the other hand, surprisingly perhaps, the positivity of the fluxes on the horizon and null infinity and the implied \emph{monotonicity} of the canonical energy with respect to a foliation of spacelike slices can already be exploited to establish a weak form of \emph{instability} in classes of spacetimes for which instability can be expected. The basic idea is that if one can find a perturbation of negative initial canonical energy, then the energy will remain negative by monotonicity, which in turn prevents future convergence to zero (or rather, to a pure gauge solution) of the perturbation. This insight goes back to \cite{Friedman, HollandsWald} and has been exploited subsequently in \cite{GreenWald} where a form of superradiant instability is established for linear perturbations in the context of asymptotically anti-de Sitter black hole spacetimes in all dimensions.\footnote{See also \cite{Dold} for an explicit construction of exponentially growing modes in the context of the massive wave equation on $4$-dimensional Kerr-AdS spacetimes.} Moreover, from the negativity of the canonical energy one can, for a particular class of perturbations, deduce in fact exponential growth of those perturbations \cite{Prabhu}.

Returning to the case of perturbations on Schwarzschild, we emphasise that the conservation laws presented in this paper will a priori also fail to produce manifestly non-negative energy fluxes on general null hypersurfaces, similar to the canonical energy on general spacelike slices discussed above.\footnote{It is an interesting question whether the conservation law satisfied by the canonical energy of \cite{HollandsWald} can be directly related to the two conservation laws presented in this paper.} Our main achievement here -- besides \emph{stating} a conservation law \emph{in a geometric form} based on a double null foliation-- is therefore to establish certain positivity properties of the fluxes appearing in it by exploiting the gauge invariance (up to boundary terms) of the fluxes, so that the a priori control promised by Theorem \ref{theo:mtintro} can indeed be deduced. Our argument proceeds along the following lines:
\begin{enumerate}
\item The conservation law (cf.~Section \ref{sec:bcl}) is expressed entirely in terms of fluxes through null hypersurfaces, the latter having the advantage that positive definiteness of fluxes is often easier to establish than on spacelike slices.
\item The fluxes appearing are shown to be invariant up to boundary terms (on round $2$-spheres) under the addition (or subtraction) of a large class of pure gauge solutions (PGS); see Section \ref{sec:ginv}.
\item Subtraction of a suitably chosen PGS, normalised to a fixed outgoing null hypersurface $C_{u_{fin}}$, can be exploited to demonstrate positive definiteness of the flux on that hypersurface \emph{up to a boundary term}. Moreover, the linearised shear appearing in this flux is invariant under the subtraction of the PGS and hence controlled in $L^2$. See Section \ref{sec:gaugechoice}.
\item Taking the limit of the ingoing hypersurface to null infinity $\mathcal{I}^+$ shows positive definiteness of the corresponding flux (which asymptotes to the linearised shear in $L^2$) up to a boundary term which cancels the boundary term in $3.$ See Sections \ref{sec:ni} and the proof in Section \ref{sec:maintheo}.
\end{enumerate}

The main result emerging from these steps is Theorem \ref{theo:mtheo} below. See also directly Corollary \ref{cor:mtheo}. 

For completeness, we end the paper by stating a second, independent, conservation law for the system of gravitational perturbations together with the transformation properties of the corresponding fluxes, cf.~Section \ref{sec:2con}. Applications of the two conservation laws to the stability problem will be presented elsewhere.

Finally, we expect the conservation laws presented in this paper to generalise to the case of gravitational perturbations on Kerr and their positivity properties to persist at least for axisymmetric perturbations. See \cite{Dain} for a discussion of coercive conservation laws in the context of axisymmetric perturbations on \emph{extremal} Kerr spacetimes. Note however the presence of the Aretakis instability \cite{Aretakis2, Aretakis3} on the event horizon in the extremal case.

. 


\section{The system of gravitational perturbations}
In this section, we will write out the system of gravitational perturbations in the double null formulation as derived in \cite{DHR}.  We begin by recalling a few notational conventions and introduce the differential operators on the Schwarzschild manifold that will appear in the equations in Section \ref{sec:prelim}. In Section \ref{sec:dynq} we introduce the dynamical quantities of the system and finally collect all equations of the system of gravitational perturbations in Sections \ref{=lmc}--\ref{=lcc}. The last subsection, Section \ref{sec:pg}, introduces a class of exact solutions to the system, \emph{pure gauge solutions}, corresponding to infinitesimal diffeomorphisms of the non-linear theory. These will play an important role later.

\subsection{Preliminaries} \label{sec:prelim}
Let $(\tilde{\mathcal{M}},g)$ denote the maximally extended Schwarzschild spacetime. We denote by $\mathcal{M} \subset \tilde{\mathcal{M}}$ the manifold with boundary which in Kruskal coordinates \cite{Wald} corresponds to $\left(0, U_0\right] \times \left[V_0,\infty\right) \times S^2_{U,V}$ with boundary being the event horizon $\mathcal{H}^+=\{0\} \times \left[V_0,\infty\right) \times S^2_{0,V}$. We will consider the linearised Einstein equations on $\left(\mathcal{M},g\right)$ with characteristic initial data defined on the null cones $U=U_0$ and $V=V_0$. 

As carried out explicitly in Section 4 of \cite{DHR} we define from the Kruskal coordinate system a system of  double null Eddington-Finkelstein (EF) coordinates in the interior of $\mathcal{M}$ via the relation $U=\exp\left(-u/2M\right)$ and $V=\exp\left(v/2M\right)$, determining in particular $u_0$ and $v_0$. See the figure below. In the interior of $\mathcal{M}$ we can write the Schwarzschild metric in EF coordinates (with $r$ given implicitly in terms of $u,v$ (in particular $\partial_v r =-\partial_u r=\Omega^2$), cf.~\cite{DHR}) as
\[
g = -4\Omega^2 du dv + r^2\left(u,v\right) \left(d\theta + \sin^2 \theta d\phi^2\right) \ \ \ \ \ \textrm{with \ $\Omega^2 := 1 - \frac{2M}{r}$}.
\]
We also recall from \cite{DHR} the null frame
\[
 e_3=\frac{1}{\Omega} \partial_u \ \ \ , \ \ \ e_4=\frac{1}{\Omega} \partial_v \ \ \ , \ \ \  e_1, e_2 \textrm{ a (local) frame on $S^2_{u,v}$}
\]
and the associated notion of $S^2_{u,v}$-tensors, specifically $S^2_{u,v}$-scalars, $S^2_{u,v}$-one-forms and symmetric traceless $S^2_{u,v}$-tensors. The linearised equations will be written as a system of equations for such tensors.

The derivative operators $\slashed{\nabla}_3, \slashed{\nabla}_4$ act on these $S^2_{u,v}$-tensors and are defined as the projection of the spacetime covariant derivative in the direction of $\frac{1}{\Omega} \partial_u$ and $\frac{1}{\Omega} \partial_v$ respectively to the tangent space of $S^2_{u,v}$. See Section 4.3.1 of \cite{DHR}. 

The derivative operator $\slashed{\nabla}$ denotes the covariant derivative on $S^2_{u,v}$ equipped with the round metric on the sphere of radius $r\left(u,v\right)$, denoted $\slashed{g}$. The following special angular operators will appear repeatedly in the system:
\[
\slashed{div} \xi := \slashed{\nabla}^A \xi_A \ \ \ , \ \ \ \slashed{curl} \xi := \slashed{\epsilon}^{AB} \slashed{\nabla}_A \xi_B \ \ \ \textrm{and} \ \ \ ( \slashed{div} \, \theta)_B := \slashed{\nabla}^A \theta_{AB}
\]
the first two mapping an $S^2_{u,v}$-one-form $\xi$ into a scalar and the last mapping a symmetric traceless tensor $\theta$ into a one-form. Here $\slashed{\epsilon}_{AB}$ denotes the components of the volume-form associated with $\left(S^2_{u,v}, \slashed{g}\right)$. We also define the operator
\[
\slashed{\mathcal{D}}_1^\star \left(h_1,h_2\right) := -\slashed{\nabla}_A h_1 + \slashed{\epsilon}_{AB} \slashed{\nabla}^B h_2
\]
mapping a pair of functions $h_1,h_2$ into an $S^2_{u,v}$ one-form. Note that its $S^2_{u,v}$-adjoint is the operator $\slashed{\mathcal{D}}_1^\star \xi := \left(\slashed{div} \xi, \slashed{curl} \xi\right)$ mapping an $S^2_{u,v}$-one-form $\xi$ into a pair of functions. Finally, we define
\[
2\slashed{\mathcal{D}}_2^\star \xi := -\slashed{\nabla}_A \xi_B -\slashed{\nabla}_B \xi_A + \slashed{g}_{AB} \left(\slashed{\nabla}^C \xi_C\right)
\]
mapping a one-form $\xi$ into a symmetric traceless tensor on $S^2_{u,v}$. Note that $\slashed{\mathcal{D}}_2^\star$ is the adjoint of $\slashed{div}$.

\subsection{The dynamical quantities} \label{sec:dynq}
We now recall the dynamical quantities of the system of gravitational perturbations. See again \cite{DHR} for details. The dynamical quantities are
\begin{itemize}
\item the linearised metric coefficients
\[
 \frac{\glinto}{\sqrt{\slashed{g}}} \ ,\ \Olin \ , \ \bmlin \  , \  \glinh \, ,
\]
which are two scalars, a one-form and a symmetric traceless two-tensor on $S^2_{u,v}$,
\item the linearised connection coefficients
\[
\otx \ ,  \ \olin \ , \ \olinb \ , \ \otxb \ , \ \elin \ , \ \eblin \ ,  \ \xlin \ , \ \xblin  \, ,
\]
which are four scalars, two one-forms and two symmetric traceless two-tensors on $S^2_{u,v}$,
\item the linearised curvature components
\[
\Klin \ , \ \rlin \ ,\  \slin \ ,\  \blin \ ,\  \bblin \ , \ \alin \ , \  \ablin \ 
\]
which are three scalars\footnote{Strictly speaking $\sigma$ is a two-form on $S^2_{u,v}$ but can be identified by duality with a scalar as it has to be proportional to the volume form.}, two one-forms and two symmetric traceless $S^2_{u,v}$-tensors.
\end{itemize}
 As in \cite{DHR} we will also write 
\begin{align} \label{scollect}
\mathscr{S}=\left(\, \glinh \, , \, \glinto \, , \, \Olino \, , \,  \bmlin\, , \,  \otx \, , \,  \otxb\, , \,  \xlin\, , \, \xblin\, , \,  \eblin \, , \,  \elin \, , \, \olin \, , \,  \olinb \, , \,  \alin \, , \,  \blin \, , \,  \rlin \, , \,  \slin \, , \,  \bblin \, , \,  \ablin \, , \, \Klin \right)
\end{align}
to denote the collection of all dynamical quantities and refer to a solution of the system of gravitational perturbations by $\mathscr{S}$.
Furthermore, linearised metric and connection coefficients will sometimes be denoted collectively by $\Gamma$, linearised curvature components by $R$.

\begin{remark} \label{rem:smooth}
Because the above linearised quantities are defined with respect to a null frame which is not regular on the event horizon \cite{DHR}, the following \underline{weighted} linearised quantities can be shown to extend smoothly to the event horizon $\mathcal{H}^+$, cf.~in particular Sections 5.1.3 and 5.1.4 of \cite{DHR}:
\begin{gather}\label{regq1}
\glinh \, , \,   \glinto \, , \,  \bmlin \, , \,  \Olino \, , \, \otx \, , \,  \Omega^{-2} \otxb, \Omega \xlin \, , \,  \Omega^{-1} \xblin\, , \,  \elin\, , \,  \eblin\, , \,  \olin \, , \,  \Omega^{-2} \olinb \, , \,
\Omega^2 \alin\, , \,   \Omega \blin\, , \,  \rlin\, , \, \slin\, , \,  \Omega^{-1} \bblin\, , \,  \Omega^{-2} \ablin\, , \,  \Klin. 
\end{gather}\end{remark}
In the following subsections we collect the system of equations satisfied by the dynamical quantities introduced above. These are taken directly from \cite{DHR}, where they are also formally derived by explicit linearisation of the vacuum Einstein equations,
\[
Ric \left[\boldsymbol{g}\right] = 0 \, .
\]
Here we only collect the Schwarzschild background values appearing in the equations below:
\begin{align}
\Omega^2 = 1- \frac{2M}{r} \ \  ,  \ \  tr \chi = - tr \underline{\chi} = \frac{2}{r}\Omega \ \ , \ \ \omega=\Omega \hat{\omega} = -\underline{\omega} =-\Omega \hat{\underline{\omega}} = \frac{M}{r^2} \ \ , \ \ \rho = - \frac{2M}{r^3} \ \  ,  \ \ K = \frac{1}{r^2}. 
\end{align}
\subsection{Equations for the linearised metric components}
\label{=lmc}
%
The following equations hold for the linearised metric components, $\glinto \, , \, \glinh \, , \, \bmlin \, , \, \Olin$:
\begin{align} \label{stos} 
\partial_u \left(\frac{\glinto}{\sqrt{\slashed{g}}}\right)  = \otxb
\qquad , \qquad 
\partial_v \left(\frac{\glinto}{\sqrt{\slashed{g}}}\right) = \otx -  \slashed{div}\, \bmlin  \, ,
\end{align}
\begin{align} \label{stos2}
\sqrt{\slashed{g}}\, \partial_u \left( \frac{\glinh_{AB}}{\sqrt{\slashed{g}}} \right)   =2\Omega\, \xblin_{AB}
\ \ \ ,  \ \ \
\sqrt{\slashed{g}}\, \partial_v \left( \frac{\glinh_{AB}}{\sqrt{\slashed{g}}} \right) &=2\Omega\, \xlin_{AB} + 2\left(\slashed{\mathcal{D}}_2^\star \bmlin \right)_{AB}  ,
\end{align}
\begin{align} \label{bequat} 
\partial_u \bmlin^A = 2 \Omega^2\left(\elin^A - \eblin^A\right) \, ,
\end{align}
\begin{align} \label{oml3}
\partial_v \left( \Olin \right) = \olin \textrm{ \ , \ }  \partial_u \left(\Olin\right)=\olinb \textrm{ \  , \ }   2 \slashed{\nabla}_A \left(\Olin\right) = \elin_A + \eblin_A.
\end{align}

\subsection{Equations for the linearised Ricci coefficients}
\label{=lRc}
For $\otx \, , \, \otxb$ we have the equations
\begin{align} \label{dtcb}
\partial_v \otxb  = \Omega^2 \left( 2 \slashed{div}\, \eblin + 2\rlin + 4 \rho \, \Olin \right) - \frac{1}{2}  \Omega tr \chi \left( \otxb - \otx  \right) ,
\end{align}
\begin{align} \label{dbtc}
\partial_u \otx  = \Omega^2 \left( 2 \slashed{div}\, {\elin} + 2 \rlin + 4\rho \, \Olin \right) - \frac{1}{2}  \Omega tr \chi \left( \otxb - \otx  \right) ,
\end{align}
\begin{align} \label{uray}
\partial_v \otx = - \left(\Omega tr \chi\right)\otx + 2 \omega \otx  + 2  \left(\Omega tr \chi \right) \olin ,
\end{align}
\begin{align} \label{vray}
\partial_u \otxb = - \left(\Omega tr \underline{\chi}\right) \otxb  + 2 \underline{\omega} \otxb + 2  \left(\Omega tr \underline{\chi} \right) \olinb ,
\end{align}
while for $\xlin \, , \, \xblin$ we have
\begin{equation} \label{tchi} 
\begin{split}
\slashed{\nabla}_3  \left(\Omega^{-1} \xblin  \right)  +  \Omega^{-1} \left(tr \underline{\chi}\right) \xblin = -\Omega^{-1} \ablin \, , \\
\slashed{\nabla}_4  \left(\Omega^{-1} \xlin \right)  +  \Omega^{-1} \left(tr{\chi}\right) \xlin= -\Omega^{-1} \alin   \, ,
\end{split}
\end{equation}
\begin{align} \label{chih3}
\slashed{\nabla}_3  \left(\Omega \xlin \right)  + \frac{1}{2} \left(\Omega tr \underline{\chi}\right) \xlin + \frac{1}{2} \left( \Omega tr \chi\right) \xblin  &= -2 \Omega \slashed{\mathcal{D}}_2^\star \elin \, , \\
\slashed{\nabla}_4  \left(\Omega \xblin  \right) + \frac{1}{2} \left(\Omega tr \chi \right) \xblin  + \frac{1}{2} \left( \Omega tr \underline{\chi}\right) \xlin &= -2 \Omega \slashed{\mathcal{D}}_2^\star \eblin \, . \label{chih3b}
\end{align}
We also have the (purely elliptic) linearised Codazzi equations on the spheres $S^2_{u,v}$, which read
\begin{equation}
\begin{split}
\slashed{div} \xblin = -\frac{1}{2} \left(tr \underline{\chi}\right)  \elin + \bblin + \frac{1}{2\Omega} \slashed{\nabla}_A \otxb , \\
\slashed{div} \xlin = -\frac{1}{2} \left( tr {\chi}\right) \eblin  -\blin + \frac{1}{2\Omega} \slashed{\nabla}_A \otx \label{ellipchi} \, .
\end{split}
\end{equation}
For $\elin$ and $\eblin$ we have the transport equations
\begin{align} \label{propeta}
\slashed{\nabla}_3 \eblin =  \frac{1}{2} \left(tr \underline{\chi}\right) \left( \elin - \eblin\right)  + \bblin
\textrm{ \ \ \ \ , \ \ \ \ }
\slashed{\nabla}_4 \elin =  -  \frac{1}{2} \left( tr {\chi}\right) \left( \elin - \eblin\right) - \blin ,
\end{align}
together with the elliptic equations on the spheres $S^2_{u,v}$
\begin{align} \label{curleta}
\slashed{curl} \elin = \slin \ \ \ , \ \ \ \slashed{curl} \eblin = -\slin \, .
\end{align}
We finally have the transport equations for $\olin$ and $\olinb$ 
\begin{align} \label{oml1}
\partial_v \olinb = -\Omega \left(\rlin + 2 \rho \Olin \right) \, ,
\end{align}
\begin{align} \label{oml2}
\partial_u \olin = -\Omega \left(\rlin + 2 \rho \Olin \right) \, ,
\end{align}
and the linearised Gauss equation on the spheres $S^2_{u,v}$, which reads
\begin{equation} \label{lingauss}
\Klin = -\rlin - \frac{1}{4} \frac{tr {\chi}}{\Omega}\left( \otxb - \otx  \right) +\frac{1}{2}\Olin \left(tr \chi tr \underline{\chi} \right) \, .
\end{equation}

\subsection{Equations for linearised curvature components}
\label{=lcc}
We finally collect the equations satisfied by the linearised curvature components, which arise from the linearisation of the Bianchi equations:
\begin{align}
\slashed{\nabla}_3 \alin + \frac{1}{2} tr \underline{\chi}\alin + 2 \underline{\hat{\omega}} \alin &= -2 \slashed{\mathcal{D}}_2^\star \blin - 3 \rho\, \xlin \, ,  \label{Bianchi1} \\
\slashed{\nabla}_4 \blin + 2 (tr \chi) \blin - \hat{\omega} \blin &= \slashed{div}\, \alin \, , \label{Bianchi2} \\
\slashed{\nabla}_3 \blin + (tr \underline{\chi}) \blin + \underline{\hat{\omega}} \blin &= \slashed{\mathcal{D}}_1^\star \left(-\rlin \, , \, \slin \, \right) + 3\rho \, \elin \, ,   \label{Bianchi3}
\end{align}
\begin{align}
\slashed{\nabla}_4 \rlin + \frac{3}{2} (tr \chi) \rlin = \slashed{div}\, \blin - \frac{3}{2} \frac{\rho}{\Omega}  \otx \, , \label{Bianchi4}
\end{align}
\begin{align}
\slashed{\nabla}_3 \rlin + \frac{3}{2} (tr \underline{\chi}) \rlin = -\slashed{div}\, \bblin - \frac{3}{2} \frac{\rho}{\Omega} \otxb \, , \label{Bianchi5}
\end{align}
\begin{align}
\slashed{\nabla}_4 \slin + \frac{3}{2} (tr \chi) \slin&= -\slashed{curl}\, \blin \, , \label{Bianchi6} \\
\slashed{\nabla}_3 \slin + \frac{3}{2} (tr \underline{\chi}) \slin &= -\slashed{curl}\, \bblin \, , \label{Bianchi7} \\
\slashed{\nabla}_4 \bblin + (tr \chi)  \bblin + \hat{\omega} \bblin &= \slashed{\mathcal{D}}_1^\star \left(\rlin \, ,  \, \slin \, \right) - 3 \rho\, \eblin  \, ,  \label{Bianchi8} \\
\slashed{\nabla}_3 \bblin + 2 (tr \underline{\chi})  \bblin - \hat{\underline{\omega}} \bblin &= - \slashed{div}\, \ablin \, , \label{Bianchi9} \\
\slashed{\nabla}_4 \ablin + \frac{1}{2} (tr \chi) \ablin + 2 \hat{\omega} \ablin &=  2 \slashed{\mathcal{D}}_2^\star \bblin - 3 \rho\, \xblin \, .  \label{Bianchi10}
\end{align}

\subsection{A class of pure gauge solutions} \label{sec:pg}
There exist special solutions to the system of gravitational perturbations above which correspond to infinitesimal coordinate transformations preserving the double null form of the metric. These are called \emph{pure gauge solutions} of the system of gravitational perturbations. A particular subset of them is identified in the following Lemma, which is proven as Lemma 6.1.1 of \cite{DHR}. Recall the notation $\Delta_{S^2}=r^2 \slashed{\Delta}$, so $\Delta_{S^2}$ is the Laplacian on the unit sphere with metric $\gamma$.

\begin{lemma} \label{lem:exactsol} 
For any smooth function $f=f\left(v,\theta,\phi\right)$, the following is a pure gauge solution of the system of gravitational perturbations:
\begin{align}
2\Olin &= \frac{1}{\Omega^2} \partial_v \left(f \Omega^2\right)  , & \glinh&= - \frac{4}{r} r^2 \slashed{\mathcal{D}}_2^\star \slashed{\nabla}_A f  \ , &\frac{\glinto}{\sqrt{\slashed{g}}} &= \frac{2\Omega^2 f}{r} + \frac{2}{r} r^2 \slashed{\Delta}f  , \nonumber \\
\bmlin &= -2r^2 \slashed{\nabla}_A \left[ \partial_v \left(\frac{f}{r}\right)\right] , & \elin &= \frac{\Omega^2}{r^2} r \slashed{\nabla} f  , & \eblin &= \frac{1}{\Omega^2} r \slashed{\nabla} \left[\partial_v \left(\frac{\Omega^2}{r}f\right) \right]   \nonumber \, ,
\nonumber \\
\xblin &= -2\frac{\Omega}{r^2} r^2 \slashed{\mathcal{D}}_2^\star \slashed{\nabla} f  , & \otx &= 2 \partial_v \left(\frac{f \Omega^2}{r}\right)  , & \otxb &=  2\frac{\Omega^2}{r^2} \left[\Delta_{\mathbb{S}^2} f - f \left(1-2\Omega^2\right) \right]  , \nonumber \\
\rlin &= \frac{6M \Omega^2}{r^4} f  , & \bblin&= \frac{6M\Omega}{r^4}  r \slashed{\nabla} f , & \Klin &= -\frac{\Omega^2}{r^3}\left(\Delta_{\mathbb{S}^2} f + 2f\right) \nonumber
\end{align}
and
\[
\xlin = \alin = \ablin = 0 \ \ \ , \ \ \ \blin = 0 \ \ \ , \ \ \ \slin = 0 \nonumber \, .
\]
We will call $f$ a gauge function. 
\end{lemma}

\section{The class of solutions} \label{sec:cls}
In \cite{DHR} we discussed the characteristic initial value problem for the above system of gravitational perturbations.  In particular, a notion of smooth characteristic seed initial data was defined to which a unique smooth solution of the system was associated, cf.~Theorem 8.1 of \cite{DHR}. We also defined a notion of asymptotically flat seed initial data. When we talk about a solution $\mathscr{S}$ below, we will always mean a smooth solution arising from the characteristic future initial value problem posed on two null cones
\[
C_{u_0} =\{ u_0 \} \times \{v \geq v_0\} \times S^2  \ \ \ \textrm{and} \ \ \ C_{v_0} =\{ u \geq u_0 \} \times \{v \} \times S^2
\]
as illustrated by the figure in the introduction. The restriction to smooth is of course not needed but convenient in the considerations below.

In \cite{DHR} we exhibited a $4$-parameter family of explicit solutions to the system of gravitational perturbations, the linearised Kerr solutions, denoted collectively by $\mathscr{K}$. These solutions are supported only on the spherical harmonics $\ell=0,1$. We proved in Theorem 9.2 of \cite{DHR} that any solution of the system supported only on the harmonics $\ell=0,1$ is equal to the sum of (a member of) $\mathscr{K}$ and an (explicit) pure gauge solution. This fact allows one to restrict to solutions supported on $\ell\geq 2$:

\begin{definition}
Let $\mathscr{S}$ be a solution of the system of gravitational perturbations. We say that $\mathscr{S}$ is supported on $\ell \geq 2$, if any scalar quantity of $\mathscr{S}$ has vanishing projection to the $\ell=0$ and $\ell=1$ spherical harmonics (see \cite{DHR}) and if any one-form $\xi$ of $\mathscr{S}$ satisfies that the scalars $\slashed{div} \xi$ and $\slashed{curl} \xi$ have vanishing projection to the $\ell=1$ spherical harmonic.
\end{definition}

Adding pure gauge solutions $\mathscr{G}$ to a given solution $\mathscr{S}$ can moreover be used to achieve certain gauge conditions on the initial data (``normalise the data"). We recall them below. Compared with \cite{DHR} we distinguish here between ``partially" initial data normalised and ``fully" initial data normalised solutions supported on $\ell \geq 2$. This is merely to state the theorem of this paper with minimal assumptions.


\begin{definition}
We call $\mathscr{S}$ a {\bf partially initial data normalised solution supported on $\ell \geq 2$} of the system of gravitational perturbations if $\mathscr{S}$ is supported on $\ell \geq 2$ and the initial data satisfies
\begin{enumerate}
\item The horizon gauge conditions, i.e.
\begin{align} \label{hgc}
\otx \left(\infty,v_0,\theta,\phi\right)=0  \ \ \ \textrm{and} \ \ \ \left(\slashed{div} \elin + \rlin \right) \left(\infty,v_0,\theta,\phi\right)=0 \, . 
\end{align} 
\item The basic round sphere condition at infinity, i.e.~the linearised Gauss-curvature satisfies
\begin{align}
\lim_{v \rightarrow \infty} r^2 \Klin \left(u_0,v,\theta,\phi\right) = 0  \ \ \ \  \textrm{along the null hypersurface $C_{u_0}$.} \label{rsc}
\end{align}
\end{enumerate}
\end{definition}

Note that both horizon gauge conditions are evolutionary, i.e.~a solution $\mathscr{S}$ satisfying (\ref{hgc}) satisfies $\otx=0$ and $\slashed{div} \elin + \rlin=0$ along \emph{all} of $\mathcal{H}^+$. This follows directly from the transport equations for these quantities along $\mathcal{H}^+$. Similarly, the round sphere condition can be seen to be evolutionary, i.e.~$r^2K^{(1)}$ vanishes as $v\rightarrow \infty$ along any cone $C_u$ with $u_0 \leq u<\infty$. See Proposition 9.4.1 and Corollary A.1 of \cite{DHR}.

\begin{definition}
We call $\mathscr{S}$ a {\bf fully initial data normalised solution supported on $\ell \geq 2$} if it is partially initial data normalised supported on $\ell \geq 2$ and if the initial data of $\mathscr{S}$ satisfy in addition
\begin{enumerate}
\item The second round sphere condition at infinity, i.e.
\begin{align}
\lim_{v \rightarrow \infty} r^2 \slashed{\mathcal{D}}_2^\star \slashed{\mathcal{D}}_2 \glinh\left(u_0,v,\theta,\phi\right) = 0  \ \ \ \  \textrm{holds along the null hypersurface $C_{u_0}$.} \label{rscm}
\end{align}
\item The lapse and shift gauge condition, i.e.
\begin{align} \label{lapsegc}
\Olin = w\left(\theta,\phi\right) \ \ \ \ \textrm{holds along both $C_{u_0}$ and $C_{v_0}$}
\end{align}
for a smooth function $w\left(\theta,\phi\right)$ on the unit sphere which has vanishing projection to $\ell=0,1$, 
and
\[
\bmlin=0 \textrm{ \ \ \ \  holds along $C_{u_0}$.}
\]
\end{enumerate}
\end{definition}
Partially and fully initial data normalised solutions supported on $\ell \geq 2$ will typically be denoted by $\Si^\prime$ consistent with the notation in \cite{DHR}.

Below we also want to take certain limits of the solution as $v \rightarrow \infty$ for fixed $u$ (or a fixed bounded subset of $u$ values). For this, the following definition is convenient:
\begin{definition} \label{def:extendskri}
We call a solution $\mathscr{S}$ {\bf extendible to null infinity} if the following weighted quantities of $\mathscr{S}$ have well-defined finite limits on null infinity\footnote{In particular, for every $u \geq u_0$ fixed the limit of the quantity along the null cone $C_{u}$ as $v \rightarrow \infty$ is well-defined.} for some $0<s<1$
\begin{equation}\label{decreten} 
\begin{split}
r^{3+s} \alin \ , \ r^{3+s}\blin \ , \  r^3\rlin  , \  r^3 \slin  , \  r^2 \bblin \ , \  r\ablin \ , \ r^3 \Klin \, ,  \\
 r^2 \xlin \ , \  r\xblin \ , \ r \elin \ , \  r^2 \eblin \ , \  r^2 \slashed{div} \elin \, , \, r^3 \slashed{div} \eblin  \ , \  r^{2+s} \olin \ , \  \olinb \ , \   r^2 \otx  \ , \  r \otxb  \ , \ \Olin  \, .
\end{split}
\end{equation}
In addition, denoting an arbitrary representative of the quantities in (\ref{decreten}) by $\mathcal{Q}$, for any fixed $u_{fin}$ with $u_0<u_{fin}<\infty$ the estimate
\begin{align}
\sup_{\left[u_0,u_{fin}\right] \times \{ v\geq v_0\} \times S^2} |\mathcal{Q}| \leq C\left[u_{fin}\right]
\end{align}
holds with the constant $C\left[u_{fin}\right]$ depending only on $u_{fin}$ (and the initial data) but not on $v$.
\end{definition}

In Theorems 9.1 and 9.2 of \cite{DHR} and Theorem A.1 of the appendix of \cite{DHR} we proved the following statement, which expresses the fact that there is no restriction in considering {\bf fully initial data normalised solutions supported on $\ell \geq 2$ which are extendible to null infinity}:

\begin{theorem*}[\cite{DHR}]
Given \underline{any} asymptotically flat to order $n\geq 12$ smooth characteristic seed initial data set (see Definition 8.2 of \cite{DHR}) with corresponding solution $\mathscr{S}$, we can construct a pure gauge solution $\mathscr{G}$ of the system of gravitational perturbations and a linearised Kerr solution $\mathscr{K}$, both explicitly computable and controllable from the seed data, such that $\Si^\prime=\mathscr{S}+\mathscr{G}-\mathscr{K}$ is a fully initial data normalised solution supported on $\ell \geq 2$ which is extendible to null infinity.
\end{theorem*}

Our distinction between partially and fully initial data normalised above emphasises that we will only exploit the validity of a subset of the gauge conditions of \cite{DHR} on the solution to obtain our main theorem. We note in particular that both the initial data normalised solution $\Si^\prime$ and the horizon-renormalised solution $\Sf^\prime$ of \cite{DHR} are partially initial data normalised in the language of this paper.

We end this section with the following simple observation (see the quantity $\Zlin$ in \cite{DHR}):

\begin{proposition} \label{prop:idc}
Consider $\Si^\prime$ a partially initial data normalised solution supported on $\ell \geq 2$. As $u \rightarrow \infty$ along the initial cone $C_{v_0}$ we have
\[
r \otx -4 \Omega^2 \Olin \left(v_0,u,\theta,\phi\right) =  \mathcal{O}\left(\Omega^4\right) 
\]
and also the angular commuted version
\[
r^2 \slashed{\Delta} \left( r \otx -4 \Omega^2 \Olin \left(v_0,u,\theta,\phi\right)\right) =  \mathcal{O}\left(\Omega^4\right) \, .
\]
\end{proposition}

\begin{proof}
Equations (\ref{dbtc}), (\ref{oml2}) along $C_{v_0}$ yield
\[
\partial_u \left( r\otx - 4M \Omega^2 \Olin\right) = -\frac{8M}{r^2} \Omega^2 \Olin  +\frac{8M}{r^2} \Omega^2 \Olin +  \mathcal{O}\left(\Omega^4\right)  \, ,
\]
where we have used the horizon gauge conditions and the smoothness of the solution at the horizon (cf.~Remark \ref{rem:smooth}). By the horizon gauge condition (\ref{hgc}), the quantity in brackets vanishes on the horizon and hence integration yields the desired result. The second claim follows by trivial commutation and the smoothness of the solution.
\end{proof}

\section{The basic conservation law} \label{sec:bcl}
In this section we define, for a smooth solution of the system of gravitational perturbations, basic energy fluxes on null hypersurfaces (Section \ref{sec:eflux}). For $\mathscr{S}$ a partially initial data normalised solution supported on $\ell \geq 2$ that is extendible to null infinity we explicitly obtain the limiting fluxes on the horizon and null infinity. Finally, we state and prove a conservation law relating the fluxes (Section \ref{sec:cl}). 

\subsection{The energy fluxes} \label{sec:eflux}

Let $\mathscr{S}$ be a smooth solution of the system of gravitational perturbations (cf.~Section \ref{sec:cls}). For any $u_0\leq u_1<u_2 \leq \infty$ and $v_0\leq v_1 < v_2 \leq \infty$
let us define the fluxes
\begin{align} \label{vflux}
F_v \left[\Gamma , \mathscr{S} \right] \left(u_1,u_2\right)= \int_{u_1}^{u_2} d u \int_{S^2} d\theta d\phi \ r^2 \sin \theta \Bigg[ -2 \olinb \otx - \frac{1}{2} \left( \otxb \right)^2 \nonumber \\
- \frac{4M}{r^2} \left(\Olin\right)\otxb +2\Omega^2 |\elin|^2 + \Omega^2 |\xblin|^2 \Bigg]
\end{align}
and
\begin{align} \label{uflux}
F_u \left[\Gamma, \mathscr{S} \right] \left(v_1,v_2\right)= \int_{v_1}^{v_2} d v \int_{S^2} d\theta d\phi \ r^2 \sin \theta \Bigg[ -2\olin \otxb - \frac{1}{2} \left( \otx \right)^2 \nonumber \\
+ \frac{4M}{r^2} \left(\Olin\right) \otx +2\Omega^2 |\eblin|^2 + \Omega^2 |\xlin|^2 \Bigg] \, .
\end{align}
Note that the flux (\ref{vflux}) remains well defined at the horizon, i.e.~in the limit $u_2 \rightarrow \infty$ by the smoothness of the solution at the horizon. Similarly, for a solution that extends to null infinity (cf.~Definition \ref{def:extendskri}) the flux (\ref{uflux}) remains well defined as $v_2 \rightarrow \infty$ for any fixed $u$. 

\subsubsection{The flux on the horizon}
If the solution $\mathscr{S}=\Si^\prime$ is partially initial data normalised supported on $\ell \geq 2$ we have for any fixed $v_0 \leq v_1<v<\infty$ the horizon limit
\begin{align} \label{aq}
F_{\infty}\left[\Gamma, \Si^\prime \right] \left(v_1,v\right) = \lim_{u \rightarrow \infty} F_{u}\left[\Gamma, \Si^\prime \right] \left(v_1,v\right) = \int_{v_1}^v d v \int_{S^2} d\theta d\phi \ r^2 \sin \theta \left[ \Omega^2 | \xlin|^2 \right] ,
\end{align}
since $\otx=0$ on $\mathcal{H}^+$ and the quantities $\Omega^{-2}\otxb$ and $\eblin$ are regular on the event horizon. In fact, one easily sees that the validity of the horizon gauge condition $\otx=0$ on $S^2_{\infty,v_0}$ alone for a general $\mathscr{S}$ is sufficient for concluding the limit (\ref{aq}).

\subsubsection{The flux on null infinity} \label{sec:ni}
To investigate the limiting flux on null infinity, we consider $\mathscr{S}=\Si^\prime$ a partially initial data normalised solution supported on $\ell \geq 2$ which extends to null infinity (cf.~Definition \ref{def:extendskri}). Let us fix $u_0 \leq u_1<u_2<\infty$ and a large $v$ which we will eventually send to infinity. Note that we have $v\sim r$  in $\mathcal{M} \cap \{v\geq v_0\} \cap \{u_1 \leq u \leq u_2\}$ with the constant implicit in $\sim$ depending on $u_1,u_2,v_0$.
From (\ref{lingauss}) we have
\[
r\otxb = -4 \Olin + \mathcal{O}\left(r^{-1}\right) \ \ \ \  \textrm{along any cone $C_{u}$ with $u<\infty$.}
\]

We now observe that the third term in the integrand of (\ref{vflux}) will vanish in the limit $v\rightarrow \infty$ while the first, second and fourth  will combine to a pure boundary term. The details are as follows:
\begin{align}
-2 \int_{u_1}^{u_2} d u  \int_{S^2} d\theta d\phi \ & r^2 \sin \theta \ \olinb \otx = -2\int_{S^2} \sin \theta d\theta d\phi r^2 \Olin  \otx \Bigg|^{u_2}_{u_1} \nonumber \\
&+2 \int_{u_1}^{u_2} d u \int_{S^2} d\theta d\phi \ r^2 \sin \theta \ \Olin \left(2 \slashed{div} \elin - \frac{1}{r}\otxb + \mathcal{O} \left(r^{-3}\right)\right) \, , \nonumber
\end{align}
where we have used the transport equation (\ref{dbtc}) and recalled $\otx \sim r^{-2}$ from Definition \ref{def:extendskri}.
By the same definition $\eblin \sim r^{-2}$ and hence we see that in the limit
\begin{align} \label{niflux}
\lim_{v \rightarrow \infty} F_{v} \left[\Gamma, \Si^\prime \right] \left(u_1,u_2\right) = & \lim_{v \rightarrow \infty} \int_{u_1}^{u_2} d u \int_{S^2} d\theta d\phi \ r^2 \sin \theta   |\xblin|^2 \left(u,v\right) \nonumber \\
+ &\lim_{v\rightarrow \infty} \frac{1}{2}\int_{S^2} \sin \theta d\theta d\phi r^3 \otxb \otx  \left(u,v\right)\Bigg|^{u_2}_{u_1}  \, .
\end{align}
This is reminiscent of the Bondi mass loss formula, with the first term on the right hand side representing the flux of gravitational energy between retarded times $u_1$ and $u_2$. We summarise the above as
\begin{proposition} \label{prop:nullit}
Let $\Si^\prime$ be a partially initial data normalised solution supported on $\ell \geq 2$ that extends to null infinity (Definition \ref{def:extendskri}). Then for $u_0 \leq u_1<u_2<\infty$ the flux (\ref{vflux}) satisfies (\ref{niflux}) in the limit on null infinity.
\end{proposition}

\subsection{The conservation law} \label{sec:cl}
The following conservation law holds:
\begin{proposition} \label{prop:conslaw}
Let $\mathscr{S}$ be a solution to the system of gravitational perturbations. \\
For any $u_0<u_1< u_2\leq \infty$ and $v_0 < v_1<v_2 < \infty$ we have the conservation law
\begin{align}
F_v\left[\Gamma, \mathscr{S} \right] \left(u_0,u_1\right)+ F_u \left[\Gamma, \mathscr{S} \right] \left(v_0,v_1\right) = F_{v_0} \left[\Gamma, \mathscr{S} \right] \left(u_0,u_1\right) +  F_{u_0}\left[\Gamma, \mathscr{S} \right] \left(v_0,v_1\right) \, .
\end{align}
\end{proposition}

\begin{proof}
Direct computation using the null structure and Codazzi equations. We compute 
\begin{align}
\partial_v \left( r^2  \Big[ -2 \olinb \otx - \frac{1}{2} \left( \otxb \right)^2 
 - \frac{4M}{r^2} \left(\Olin\right)\otxb +2\Omega^2 |\elin|^2 + \Omega^2 |\xblin|^2 \Big] \right) \nonumber \\
=-2 \left(r^2 \otx\right) \left(-\Omega^2 \left(\rlin -\frac{4M}{r^3} \Olin \right)\right) - 2 \olinb r^2 \left(+ \frac{2M}{r^2} \otx +\frac{4\Omega^2}{r} \olin \right) \nonumber \\
- r\otxb \cdot r \left( \Omega^2 \left[ 2 \slashed{div} \underline{\elin} + 2 \rlin -\frac{8M}{r^3} \Olin \right] + \frac{\Omega^2}{r}   \otx  \right) 
-4M \cdot \olin \otxb \nonumber \\
 -4M \left(\Olin\right) \left( \Omega^2 \left[ 2 \slashed{div} \underline{\elin} + 2 \rlin -\frac{8M}{r^3} \Olin \right] - \frac{\Omega^2}{r} \left(\otxb - \otx  \right) \right) \nonumber \\
 +4M \Omega^2 |\elin|^2 + 4r \Omega^2\elin \left( + \Omega^2 \underline{\elin} - \Omega r \blin\right) 
 + 2r \Omega\xblin \cdot \Omega \left(\Omega^2 \xlin - 2\Omega r \slashed{\mathcal{D}}_2^\star \underline{\elin} \right) \, ,
\end{align}
and similarly
\begin{align}
\partial_u \left(r^2 \Big[ -2\olin \otxb - \frac{1}{2} \left( \otx \right)^2
+ \frac{4M}{r^2} \left(\Olin\right) \otx +2\Omega^2 |\eblin|^2 + \Omega^2 |\xlin|^2 \Big]\right) \nonumber \\
=-2 \left(r^2\otxb\right) \left(-\Omega^2 \left(\rlin -\frac{4M}{r^3} \Olin \right)\right) - 2 \olin r^2 \left(-\frac{2M}{r^2} \otxb -\frac{4\Omega^2}{r} \olinb \right) \nonumber \\
- r\otx \cdot r \left(\Omega^2 \left[ 2 \slashed{div} \elin + 2 \rlin -\frac{8M}{r^3} \Olin \right] -  \frac{\Omega^2}{r}\otxb  \right) + 4M \cdot \olinb \otx \nonumber \\
+ 4M \left( \Olin \right) \left(\Omega^2 \left[ 2 \slashed{div} {\elin} + 2 \rlin -\frac{8M}{r^3} \Olin \right] - \frac{\Omega^2}{r} \left(\otxb - \otx  \right)\right)\nonumber \\
 -4M \Omega^2 |\eblin|^2 + 4r \Omega^2\eblin \left( - \Omega^2 \elin + \Omega r \bblin\right) 
 + 2r \Omega \xlin \cdot \Omega \left(-\Omega^2\xblin - 2\Omega r \slashed{\mathcal{D}}_2^\star {\eta} \right) \, . \nonumber
\end{align}
Summing the two expressions and integrating over $\int_{S^2} \sin \theta d\theta d\phi$ (which we do not write out, instead ``$\equiv$" indicates equality after this integration) we find the expression
\begin{align}
\equiv -2 r^2 \otx \left(-\Omega^2 \left(\rlin -\frac{4M}{r^3} \Olin \right)\right) - 2 \olinb r^2 \left(+ \frac{2M}{r^2} \otx +\frac{4\Omega^2}{r} \olin \right) \nonumber \\
- 2r^2\otxb \cdot \left( \Omega^2 \left[  \slashed{div} \eblin +  \rlin -\frac{4M}{r^3} \Olin \right]  \right) 
-4M \cdot \olin \otxb \nonumber \\
 -2 r^2\otxb \left(-\Omega^2 \left(\rlin -\frac{4M}{r^3}\Olin \right)\right) - 2 \olin r^2 \left(-\frac{2M}{r^2} \otxb -\frac{4\Omega^2}{r} \olinb \right) \nonumber \\
- 2r^2\otx \cdot  \left(\Omega^2 \left[  \slashed{div} \elin +  \rlin -\frac{4M}{r^3} \Olin \right]  \right) + 4M \cdot \olinb \otx \nonumber \\
+ 4M \left( \Olin \right)\Omega^2 \left( 2 \slashed{div} \elin - 2 \slashed{div} \eblin \right) \nonumber \\
+4M \Omega^2 |\elin|^2 -4r^2 \Omega^3 \elin \left(\blin + \slashed{div} \xlin \right)  -4M \Omega^2 |\eblin|^2 
 + 4 r^2 \Omega^3 \eblin \left( \bblin - \slashed{div}\xblin \right) \, .
\end{align}
From this further cancellations in the first two lines are obvious and an additional integration by parts yields
\begin{align}
\equiv - 2 \olinb r^2 \left(+\frac{2M}{r^2}\otx + \frac{4\Omega^2}{r}\olin \right) 
- 2r^2\otxb \cdot  \Omega^2 \left[  \slashed{div} \elin \right]  
-4M \cdot \olin \otxb \nonumber \\
 - 2 \olin r^2 \left(-\frac{2M}{r^2}\otxb -\frac{4\Omega^2}{r} \olinb \right)
- 2r^2 \otx \cdot  \Omega^2 \left[  \slashed{div} \elin  \right]  + 4M \cdot \olinb \otx \nonumber \\
+ 8M \Omega^2  \Olin \left(  \slashed{div}\elin -  \slashed{div} \eblin \right) 
+4M \Omega^2 |\elin|^2 -4r^2 \Omega^3 \elin \left(\blin + \slashed{div} \xlin \right)  -4M \Omega^2 |\eblin|^2 
 + 4 r^2 \Omega^3 \eblin \left( \bblin - \slashed{div}\xblin \right) \nonumber \, ,
\end{align}
which is seen to vanish after a further integration by parts in the angular variables and inserting the linearised Codazzi equation (\ref{ellipchi}).
\end{proof}
The following figure illustrates the conservation law for a region bounded by the initial cones $C_{u_0}$ and $C_{v_0}$ and the future null cones given by $u_1=u_{fin}$ and $v_1=v_{fin}$. This is how it will be applied later.
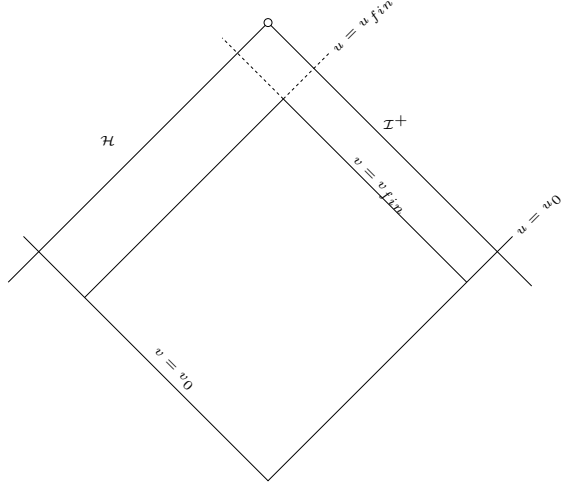
\vspace{.2cm}
\[
\begin{picture}(0,0)%
\includegraphics{consss.pstex}%
\end{picture}%
\setlength{\unitlength}{1263sp}%
\begingroup\makeatletter\ifx\SetFigFont\undefined%
\gdef\SetFigFont#1#2#3#4#5{%
  \reset@font\fontsize{#1}{#2pt}%
  \fontfamily{#3}\fontseries{#4}\fontshape{#5}%
  \selectfont}%
\fi\endgroup%
\begin{picture}(10299,10297)(2089,-9973)
\put(12151,-3961){\rotatebox{45.0}{\makebox(0,0)[lb]{\smash{{\SetFigFont{5}{6.0}{\rmdefault}{\mddefault}{\updefault}{\color[rgb]{0,0,0}$u=u_0$}%
}}}}}
\put(8551,-361){\rotatebox{45.0}{\makebox(0,0)[lb]{\smash{{\SetFigFont{5}{6.0}{\rmdefault}{\mddefault}{\updefault}{\color[rgb]{0,0,0}$u=u_{fin}$}%
}}}}}
\put(3901,-2161){\makebox(0,0)[lb]{\smash{{\SetFigFont{5}{6.0}{\rmdefault}{\mddefault}{\updefault}{\color[rgb]{0,0,0}$\mathcal{H}$}%
}}}}
\put(9451,-1861){\makebox(0,0)[lb]{\smash{{\SetFigFont{5}{6.0}{\rmdefault}{\mddefault}{\updefault}{\color[rgb]{0,0,0}$\mathcal{I}^+$}%
}}}}
\put(8851,-2461){\rotatebox{315.0}{\makebox(0,0)[lb]{\smash{{\SetFigFont{5}{6.0}{\rmdefault}{\mddefault}{\updefault}{\color[rgb]{0,0,0}$v=v_{fin}$}%
}}}}}
\put(4951,-6211){\rotatebox{315.0}{\makebox(0,0)[lb]{\smash{{\SetFigFont{5}{6.0}{\rmdefault}{\mddefault}{\updefault}{\color[rgb]{0,0,0}$v=v_0$}%
}}}}}
\end{picture}%

\]

\section{Gauge invariance of the $u$-flux modulo boundary terms} \label{sec:ginv}
In this section, we compute how the flux $F_u\left[\Gamma, \mathscr{S} \right] \left(v_1,v_2\right)$ transforms under the addition of a pure gauge solution generated by a gauge function $f\left(v, \theta,\phi\right)$. To derive these formulae we will not need any gauge conditions on the solution. We have

\begin{proposition} \label{prop:gaugechange}
Let $\mathscr{S}$ be a solution of the system of gravitational perturbations. Let $f\left(v,\theta,\phi\right)$ be a smooth gauge function generating a pure gauge solution $\mathscr{G}$ of the system of gravitational perturbations as in Lemma \ref{lem:exactsol}. Finally, set
\begin{align}
\mathscr{S} = \tilde{\mathscr{S}} + \mathscr{G}
\end{align}
thereby defining a new solution $\tilde{\mathscr{S}}$. Then the flux on fixed constant-$u$ hypersurfaces satisfies
\begin{align}
F_u \left[\Gamma, \mathscr{S} \right] \left(v_0,v\right)= F_u \left[\Gamma, \tilde{\mathscr{S}}\right] \left(v_0,v\right) + \int_{S^2} \sin \theta d\theta d\phi\left( \mathcal{G} \left(v, u,\theta,\phi\right)-\mathcal{G} \left(v_0, u,\theta,\phi\right)\right) \nonumber
\end{align}
with
\begin{align} \label{finge}
\mathcal{G}  \left(v, u,\theta,\phi\right)  = & \left(\Omega^2 f\right)^2 \frac{6M}{r^2}
 - \frac{1}{2} \Omega^{-2} r^3 \left(\otx_{\So}-{\otx}_{\tilde{\mathscr{S}}} \right)\otxb_{\So} \nonumber \\
&+ f \Omega^2 \left(2r^2 \left(-\slashed{div} \, {\eblin}_{\Sop} + {\rlin}_{\Sop}\right)\right)
- f \Omega^2 \cdot \frac{1}{r}\left(1-\frac{4M}{r}\right) r^2 \Omega^{-2} {\otx}_{\Sop}  \, ,
\end{align}
where the subscripts $\So$ or $\Sop$ indicate whether the geometric quantity is associated with the solution $\So$ or $\Sop$.
In other words, the difference of the fluxes in the old and in the new gauge is a pure boundary term.
\end{proposition}

\begin{proof}
According to Lemma \ref{lem:exactsol}, the \emph{integrand} of the $u$-flux in (\ref{uflux}) changes to
\begin{align} \label{newintegrand}
f_u\left[\Gamma, \mathscr{S}\right] := -2 \left({\olin}_{\Sop} + \frac{1}{2} \partial_v \left(\Omega^{-2} \partial_v (f\Omega^2)\right) \right) \left({\otxb}_{\Sop} + 2\frac{\Omega^2}{r^2} \left(\Delta_{S^2} f + f \left(1-\frac{4M}{r}\right) \right) \right) \nonumber \\
  - \frac{1}{2} \left( {\otx}_{\Sop} + 2 \partial_v \left(\frac{\Omega^2 f}{r}\right) \right)^2 \nonumber \\
+ \frac{4M}{r^2} \left({\Olin}_{\Sop} + \frac{1}{2} \Omega^{-2} \partial_v (f\Omega^2)\right)   \left( {\otx}_{\Sop} + 2 \partial_v \left(\frac{\Omega^2 f}{r}\right) \right) \nonumber \\
+2\Omega^2 \Big|{\eblin}_{\Sop} + \Omega^{-2} r \slashed{\nabla} \left(\partial_v \left(\frac{\Omega^2 f}{r}\right)\right) \Big|^2 +\Omega^2 |{\xlin}_{\Sop}|^2 .
\end{align}
Upon expanding there are three types of terms in (\ref{newintegrand}): The first type will produce the (integrand of the) flux in the new gauge, the second are mixed terms (``involving one $f$") and the third are the remaining terms (involving two $f$'s). As the terms of the first type are easily collected we focus on the latter two. We keep the convention that ``$\equiv$" denotes equality after integration over the unit $S^2$ (which allows integration by parts over the angular variables).
\subsubsection*{Mixed terms}
\begin{align}
\mathcal{A} \equiv - {\otxb}_{\Sop} \ \partial_v \left(\Omega^{-2} \partial_v (f\Omega^2)\right) - 4 {\olin}_{\Sop} \frac{\Omega^2}{r^2} \left(\Delta_{S^2} f + f \left(1-\frac{4M}{r}\right) \right) \nonumber \\
-2{\otx}_{\Sop} \ \partial_v \left(\frac{\Omega^2 f}{r}\right) + \frac{8M}{r^2}  \left({\Olin}_{\Sop}\right)\partial_v \left(\frac{\Omega^2 f}{r}\right) \nonumber \\
+ \frac{2M}{r^2} \Omega^{-2} \partial_v \left(f\Omega^2\right)  {\otx}_{\Sop} - 4 r \slashed{div}\, {{\eblin}_{\Sop}} \ \partial_v \left(\frac{\Omega^2 f}{r}\right) \, .
\end{align}
We denote these six terms by
\begin{align}
\mathcal{A} \equiv \mathcal{A}_1 + \mathcal{A}_2 + \mathcal{A}_3 + \mathcal{A}_4 +\mathcal{A}_5 + \mathcal{A}_6 \, .
\end{align}
We have
\begin{align} \label{a1com}
\mathcal{A}_1 \equiv &-\frac{1}{r^2} \partial_v \left(r^2 {\otxb}_{\Sop} \left(\Omega^{-2} \partial_v (f\Omega^2)\right)\right)  \nonumber \\
&+ \frac{ \partial_v (f\Omega^2)}{r}\left[2 r \slashed{div} {\eblin}_{\Sop} + 2r  {\rlin}_{\Sop}- \frac{8M}{r^2}  {\Olin}_{\Sop}  +{\otx}_{\Sop} +{\otxb}_{\Sop}\right] \, , 
\end{align}
\begin{align} \label{a2com}
\mathcal{A}_2 \equiv -\frac{4}{r^2} \partial_v \left( {\Olin}_{\Sop} \left(\Delta_{S^2} f\Omega^2 + f \Omega^2 \left(1-\frac{4M}{r}\right) \right) \right)+ \frac{16M}{r^4} f \Omega^4 {\Olin}_{\Sop} \nonumber \\
+ 2 \left(\slashed{div} \, {\elin} + \slashed{div}\, {\eblin} \right)_{\Sop} \partial_v \left(f\Omega^2\right) + \frac{4}{r^2} \left(1-\frac{4M}{r}\right) {\Olin}_{\Sop} \ \partial_v \left(f\Omega^2\right) 
\end{align}
and we can therefore write $\mathcal{A}$ as
\begin{align} 
\mathcal{A} &\equiv 
 \partial_v \left(\Omega^2 f\right) \Big[  2\slashed{div} {\elin}_{\Sop} +  2{\rlin}_{\Sop}
 -\frac{1}{r}\Omega^{-2} \left(1-\frac{4M}{r}\right) {\otx}_{\Sop}  + \frac{4}{r^2} \left(1-\frac{4M}{r}\right){\Olin}_{\Sop} 
 +\frac{1}{r} {\otxb}_{\Sop} \Big] \nonumber \\
 &+ f\Omega^2 \left[\frac{8M}{r^4}  \Omega^2 {\Olin}_{\Sop} + 2 \frac{\Omega^2}{r^2}  {\otx}_{\Sop} + \frac{4}{r} \Omega^2 \slashed{div}\, {\eblin}_{\Sop}\right] +\mathcal{B}_1 + \mathcal{B}_2  \nonumber 
\end{align}
where $\mathcal{B}_1$ and $\mathcal{B}_2$ are the boundary terms (i.e.~the first term of (\ref{a1com}) and (\ref{a2com}) respectively) encountered above.  
We now write the term multiplying $\partial_v \left(\Omega^2 f\right)$ as a boundary term:
\begin{align}
 \mathcal{A} \equiv \mathcal{B}_1 + \mathcal{B}_2 +\frac{1}{r^2} \partial_v \left(f \Omega^2 \left(\frac{2}{r} r^3 \left(\slashed{div} \, {\elin}_{\Sop} + {\rlin}_{\Sop}\right)\right) \right) \nonumber \\ + \frac{2}{r} \Omega^4 f \left(\slashed{div} \, {\elin}_{\Sop} + {\rlin}_{\Sop}\right) 
- 2\frac{\Omega^2}{r} f \left(\Omega^2 \slashed{div} {\eblin}_{\Sop} + \Omega^2 \slashed{div}\, {\elin}_{\Sop} + \frac{3M}{r^2}{\otx}_{\Sop}  \right) \nonumber \\
- \frac{1}{r^2} \partial_v \left(f \Omega^2 \cdot \frac{1}{r}\left(1-\frac{4M}{r}\right) r^2 \Omega^{-2} {\otx}_{\Sop} \right)  \nonumber \\ - f\Omega^2 \left(  \frac{1}{r^2} - \frac{8M}{r^3}\right) {\otx}_{\Sop}
+4 \left(1-\frac{4M}{r}\right) \frac{f\Omega^2}{r^2} {\olin}_{\Sop} \nonumber \\
+ 4 \frac{1}{r^2} \partial_v  \left( f\Omega^2 \left(1-\frac{4M}{r}\right){\Olin}_{\Sop} \right) \nonumber \\
- 4 \frac{1}{r^2} f\Omega^2 \left(1-\frac{4M}{r}\right){\olin}_{\Sop} - \frac{16M}{r^4}  f \Omega^4 {\Olin}_{\Sop} \nonumber \\
+ \frac{1}{r^2} \partial_v \left(f \Omega^2 \cdot r  {\otxb}_{\Sop} \right)\nonumber \\
- \frac{f\Omega^4}{r^2} \left[2 r \slashed{div}\,{\eblin}_{\Sop} + 2r {\rlin}_{\Sop} - \frac{8M}{r^2}  {\Olin}_{\Sop} + {\otx}_{\Sop} \right] \nonumber \\
+ f\Omega^4 \left[\frac{8M}{r^4}   {\Olin}_{\Sop} + 2 \frac{1}{r^2} {\otx}_{\Sop} + \frac{4}{r}  \slashed{div} \,{\eblin}_{\Sop} \right] \nonumber \, .
\end{align}
Note that all terms which are not boundary terms cancel. For the above we have used the evolution equation (which holds for both $\So$ and $\Sop$ as the pure gauge solution is a solution of the system)
\begin{align}
\partial_v \left(r^3 \left(\slashed{div} \elin + \rlin\right)\right) = \Omega^2 r^2 \slashed{div} \left( \eblin + \elin \right)+ 3M \otx
\end{align}
which is easily derived from (\ref{Bianchi4}) and (\ref{propeta}), as well as the propagation equations (\ref{uray}) and (\ref{dtcb}).
%
We summarise this as 
\begin{align}
\mathcal{A} \equiv &-\frac{1}{r^2} \partial_v \left(r^2 {\otxb}_{\Sop} \left(\Omega^{-2} \partial_v (f\Omega^2)\right)\right) \nonumber \\
 &-\frac{4}{r^2} \partial_v \left( {\Olin}_{\Sop} \left(\Delta_{S^2} f\Omega^2 + f \Omega^2 \left(1-\frac{4M}{r}\right) \right) \right) \nonumber \\
 &+\frac{1}{r^2} \partial_v \left(f \Omega^2 \left(\frac{2}{r} r^3 \left(\slashed{div} \,{\elin}_{\Sop} + {\rlin}_{\Sop}\right)\right) \right) \nonumber \\
&- \frac{1}{r^2} \partial_v \left(f \Omega^2 \cdot \frac{1}{r}\left(1-\frac{4M}{r}\right) r^2 \Omega^{-2} {\otx}_{\Sop} \right) \nonumber \\
&+ 4 \frac{1}{r^2} \partial_v  \left( f\Omega^2 \left(1-\frac{4M}{r}\right) {\Olin}_{\Sop} \right) \nonumber \\
&+ \frac{1}{r^2} \partial_v \left(f \Omega^2 \cdot r  {\otxb}_{\Sop} \right) \, ,
\end{align}
which simplifies to
\begin{align}
\mathcal{A} \equiv &-\frac{1}{r^2} \partial_v \left(r^3 {\otxb}_{\Sop} \left(\Omega^{-2} \partial_v \left(\frac{f\Omega^2}{r}\right)\right)\right) \nonumber \\
 &+\frac{1}{r^2} \partial_v \left(f \Omega^2 \left(2r^2 \left(-\slashed{div} \, {\eblin}_{\Sop}+ {\rlin}_{\Sop} \right)\right) \right) \nonumber \\
&- \frac{1}{r^2} \partial_v \left(f \Omega^2 \cdot \frac{1}{r}\left(1-\frac{4M}{r}\right) r^2 \Omega^{-2} {\otx}_{\Sop} \right) \, .
\end{align}

\subsubsection*{Remaining terms}
We turn to the remaining quadratic terms in (\ref{newintegrand}) since the terms which produce the flux expression in the new gauge are easily taken care of. Here we have
\begin{align}
 -2 \left(  \partial_v \left(\Omega^{-2} \partial_v (f\Omega^2)\right) \right) \left( \frac{\Omega^2}{r^2} \left(\Delta_{S^2} f + f \left(1-\frac{4M}{r}\right) \right) \right) 
  -  2 \left(\partial_v \left(\frac{\Omega^2 f}{r}\right) \right)^2 \nonumber \\
+ \frac{4M}{r^2} \left(\frac{1}{2} \Omega^{-2} \partial_v (f\Omega^2)\right)   \left( 2 \partial_v \left(\frac{\Omega^2 f}{r}\right) \right) 
+2\Omega^{-2} \Big| r \slashed{\nabla}_A \left(\partial_v \left(\frac{\Omega^2 f}{r}\right)\right) \Big|^2  \nonumber
\end{align}
and we will refer to this expression as $\tilde{\mathcal{A}}$ and
\begin{align}
\tilde{\mathcal{A}} = \tilde{\mathcal{A}}_1 + ... + \tilde{\mathcal{A}}_4
\end{align}
We have (recall $\gamma$ denotes the round metric on the unit sphere)
\begin{align}
\tilde{\mathcal{A}}_1 &= -\frac{2}{r^2} \partial_v \left( \Omega^{-2} \partial_v (f\Omega^2)\left(\Delta_{S^2} f\Omega^2 + f \Omega^2 \left(1-\frac{4M}{r}\right) \right) \right) \nonumber \\
&+ \frac{2}{r^2} \Omega^{-2} \partial_v \left(f \Omega^2\right) \left(\Delta_{S^2} \partial_v ( f \Omega^2)+ \partial_v (f \Omega^2) \left(1-\frac{4M}{r}\right) + f \Omega^4 \frac{4M}{r^2} \right) \nonumber \\
&\equiv-\frac{2}{r^2} \partial_v \left( \Omega^{-2}  \partial_v ( f\Omega^2)\left(\Delta_{S^2} f\Omega^2 + f \Omega^2 \left(1-\frac{4M}{r}\right) \right) \right)+\frac{4M}{r^2} \partial_v \left( \frac{1}{r^2} (f\Omega^2)^2 \right) \nonumber \\
&- \frac{2}{r^2} \Omega^{-2} |\partial_v \left(\nabla_{S^2} f \Omega^2\right)|_{\gamma}^2  + \frac{2}{r^2} \Omega^{-2} \left(1-\frac{4M}{r}\right) | \partial_v (f\Omega^2)|^2 + \frac{8M}{r^5} f^2 \Omega^6  \nonumber
\end{align}
\begin{align}
\tilde{\mathcal{A}}_2 &= -2 \frac{1}{r^2}  \left(\partial_v \left(\Omega^2 f\right) \right)^2+2 \frac{\Omega^2}{r^3} \partial_v \left(\Omega^2 f\right)^2 - 2 \frac{f^2 \Omega^8}{r^4} \nonumber \\
&=-2 \frac{1}{r^2}  \left(\partial_v \left(\Omega^2 f\right) \right)^2+2 \frac{1}{r^2} \partial_v \left(\frac{\Omega^2}{r} \left( \Omega^2 f\right)^2\right) - \frac{4M}{r^5} \Omega^6 f^2
\end{align}
\begin{align}
\tilde{\mathcal{A}}_3 &= \frac{4M}{r^3} \Omega^{-2} \left(\partial_v (f \Omega^2)\right)^2 - \frac{2M}{r^4} \partial_v \left(f\Omega^2\right)^2 \nonumber \\
&=\frac{4M}{r^3} \Omega^{-2} \left(\partial_v (f \Omega^2)\right)^2 - \frac{2M}{r^2} \partial_v \left(\frac{1}{r^2} \left(f\Omega^2\right)^2\right) - \frac{4M}{r^5} \Omega^6 f^2
\end{align}
\begin{align}
\tilde{\mathcal{A}}_4 &= +2 \frac{\Omega^{-2}}{r^2}  |\partial_v \left(\Omega^2 \nabla_{S^2} f\right)|_{\gamma}^2-2 \frac{1}{r^3} \partial_v \left|\Omega^2 \nabla_{S^2} f\right|_{\gamma}^2 + 2 \frac{|\nabla_{S^2} f|_{\gamma}^2 \Omega^6}{r^4} \nonumber \\
&= +2 \frac{\Omega^{-2}}{r^2}  |\partial_v \left(\Omega^2 \nabla_{S^2} f\right)|_{\gamma}^2-2 \frac{1}{r^2} \partial_v \left(\frac{1}{r} \left|\Omega^2 \nabla_{S^2} f\right|_{\gamma}^2\right) \, .
\end{align}
We see that all terms except boundary terms cancel and hence
\begin{align}
\tilde{\mathcal{A}} &= -\frac{2}{r^2} \partial_v \left( \Omega^{-2}  \partial_v ( f\Omega^2)\left(\Delta_{S^2} f\Omega^2 + f \Omega^2 \left(1-\frac{4M}{r}\right) \right) \right)+\frac{4M}{r^2} \partial_v \left( \frac{1}{r^2} (f\Omega^2)^2 \right) \nonumber \\
&+2 \frac{1}{r^2} \partial_v \left(\frac{\Omega^2}{r} \left( \Omega^2 f\right)^2\right) 
 - \frac{2M}{r^2} \partial_v \left(\frac{1}{r^2} \left(f\Omega^2\right)^2\right)-2 \frac{1}{r^2} \partial_v \left(\frac{1}{r} \left|\Omega^2 \nabla_{S^2} f\right|_{\gamma}^2\right)  \nonumber
\end{align}
is also a pure boundary term.

\subsubsection*{Summary}
In summary, we have proven the desired proposition for $\mathcal{G}$ being
\begin{align}
\mathcal{G}  \left(v, u,\theta,\phi\right) = -\frac{2}{r}  \left|\Omega^2 \nabla_{S^2} f\right|_{\gamma}^2 + \left(\Omega^2 f\right)^2 \left(\frac{2}{r} - \frac{2M}{r^2}\right) - \Omega^{-2} r^3 \partial_v \left(\frac{\Omega^2 f}{r}\right)\otxb_{\So} \nonumber \\
-rf \Omega^2  \left(\otxb_{\So} - {\otxb}_{\Sop} \right) + f \Omega^2 \left(2r^2 \left(-\slashed{div} {\eblin}_{\Sop} + {\rlin}_{\Sop}\right)\right) \nonumber \\
- f \Omega^2 \cdot \frac{1}{r}\left(1-\frac{4M}{r}\right) r^2 \Omega^{-2} {\otx}_{\Sop} \nonumber
\end{align}
which can be simplified to
\begin{align}
\mathcal{G}  \left(v, u,\theta,\phi\right)  = + \left(\Omega^2 f\right)^2 \left(\frac{2}{r} - \frac{2M}{r^2}\right)
 - \frac{1}{2} \Omega^{-2} r^3 \left(\otx_{\So}-{\otx}_{\Sop}\right)\otxb_{\So} \nonumber \\
-\frac{2}{r}f \Omega^4  \left(\Delta_{S^2} f + f \left(1-\frac{4M}{r}\right)\right) + f \Omega^2 \left(2r^2 \left(-\slashed{div} \, {\underline{\eta}_{\Sop}} + {\rlin}_{\Sop}\right)\right) \nonumber \\
- f \Omega^2 \cdot \frac{1}{r}\left(1-\frac{4M}{r}\right) r^2 \Omega^{-2} {\otx}_{\Sop}  -\frac{2}{r}  \left|\Omega^2 \nabla_{S^2} f\right|_{\gamma}^2\nonumber
\end{align}
and after further integration by parts in the angular variables to the expression appearing in (\ref{finge}).
\end{proof}

\section{Choice of the gauge function normalised to the cone $C_{u_{fin}}$} \label{sec:gaugechoice}
Let $\Si^\prime$ be a partially initial data normalised solution supported on $\ell \geq 2$. Fix an outgoing null cone $C_{u_{fin}}$ for some $u_{fin} > u_0$. We will now define a particular gauge function, normalised to the null cone $C_{u_{fin}}$, which will generate a pure gauge solution $\mathscr{G}$ through Lemma \ref{lem:exactsol}, which when subtracted from $\Si^\prime$ will produce positivity of the flux $F_{u_{fin}} \left[\Gamma, \Sf^\prime=\Si^\prime-\mathscr{G}\right]$. Specifically, we define the gauge function
\begin{align} \label{choicef}
f \left(v, \theta,\phi \right) = \frac{r}{\Omega^2} \left(u_{fin}, v,\theta,\phi\right) \int_{v_0}^{v} \otx_{\Si^\prime} \left(u_{fin},\bar{v}, \theta, \phi\right) d\bar{v} \, .
\end{align}
Clearly $f\left(v_0, \theta,\phi\right)=0$. We have 

\begin{lemma}  \label{lem:newval}
Under the assumptions of Proposition \ref{prop:gaugechange} with $\mathscr{S}=\Si^\prime$ a partially initial data normalised solution supported on $\ell \geq 2$ and with $f$ defined as in (\ref{choicef}) generating a pure gauge solution $\mathscr{G}$, we have on the null hypersurface $C_{u_{fin}}$ the following identities for the geometric quantities of $\Sf^\prime:=\Si^\prime - \mathscr{G}$:
\begin{align}
{\otx}_{\Sf^\prime} \left(u_{fin},v,\theta,\phi\right) = 0 \ \ \ \textrm{and} \ \ \ {\olin}_{\Sf^\prime}\left(u_{fin},v,\theta,\phi\right) = 0 \, .
\end{align}
Moreover,
\begin{align} \label{fit}
2 {\Olin}_{\Sf^\prime}\left(u_{fin},v,\theta,\phi\right)  = 2 {\Olin}_{\Sf^\prime}\left(u_{fin},v_0,\theta,\phi\right)  = \left( 2\Olin -\frac{r}{2\Omega^2} \otx\right)_{\Si^\prime} \left(u_{fin},v_0,\theta,\phi\right)
\end{align}
and
\begin{align}
r^3 \left( \slashed{div} \, {\elin}_{\Sf^\prime} + {\rlin}_{\Sf^\prime} \right) \left(u_{fin},v,\theta,\phi\right) = \ & r^3 \left( \slashed{div}\, \elin + \rlin \right)_{\Si^\prime} \left(u_{fin},v_0,\theta,\phi\right) \nonumber \\
&+ \left(r\left(u_{fin},v\right) - r\left(u_{fin},v_0\right) \right) \Delta_{S^2} \left( 2{\Olin}\right)_{\Sf^\prime} \left(u_{fin},v,\theta,\phi\right) \, .
\end{align}
\end{lemma}

\begin{proof}
Suppressing the dependence on $\theta,\phi$ in the notation for the proof, we have from Lemma \ref{lem:exactsol}
\begin{align}
\otx_{\Si^\prime} \left(u_{fin},v\right) = {\otx}_{\Sf^\prime} \left(u_{fin},v\right)+ 2 \partial_v \left(\frac{f \Omega^2}{r}\right)\left(u_{fin},v\right) \, 
\end{align}
and the first claim of the Lemma follows. From the transport equation for ${\otx}_{\Sf^\prime} \left(u_{fin},v\right)$ along $u=u_{fin}$, (\ref{uray}), we conclude that ${\olin}_{\Sf^\prime} \left(u_{fin},v\right)=0$ for any $v \geq v_0$. We also have
\begin{align}
2\Olin_{\Si^\prime} =2 {\Olin}_{\Sf^\prime} + \frac{r}{2\Omega^2} \left( \otx_{\Si^\prime} - {\otx}_{\Sf^\prime}\right) + \frac{f \Omega^2}{r} \, .
\end{align}
Therefore (using that ${\olin}_{\Sf^\prime} \left(u_{fin},v\right)=0$ for any $v \geq v_0$) , along $u=u_{fin}$ we obtain
\begin{align}
2 {\Olin}_{\Sf^\prime}\left(u_{fin},v\right)  = 2 {\Olin}_{\Sf^\prime} \left(u_{fin},v_0\right) = \left(2\Olin -\frac{r}{2\Omega^2} \otx \right)_{\Si^\prime} \left(u_{fin},v_0\right) \, .
\end{align}
To verify the last claim of the Lemma note first that
\begin{align}
\left(\slashed{div} \elin + \rlin\right)_{\Si^\prime} \left(u_{fin},v_0\right) = \left(\slashed{div} {\elin} +{\rlin}\right)_{\Sf^\prime} \left(u_{fin},v_0\right) \, 
\end{align}
since $f\left(v_0,\theta,\phi\right)=0$. Moreover, 
along $u_{fin}$ we have the evolution equation
\begin{align} \label{yu}
\partial_v \left(r^3 \left(\slashed{div} \, {\elin} + {\rlin}\right)_{\Sf^\prime} \left(u_{fin},v\right) \right) = \Omega^2\Delta_{S^2} \left(2{\Olin}\right)_{\Sf^\prime} \, .
\end{align}
Since the expression multiplying $\Omega^2$ on the right hand side is constant along $u=u_{fin}$ by the identity (\ref{fit}), integration of (\ref{yu}) yields the claim.
\end{proof}

\begin{remark} \label{rem:hozlimit}
Note that since $\Si^\prime$ is partially initial data normalised supported on $\ell \geq 2$ we have
\begin{align}
\lim_{u_{fin} \rightarrow \infty} \left(\slashed{div} \elin + \rlin\right)_{\Si^\prime} \left(u_{fin},v_0,\theta,\phi\right) = 0 \, 
\end{align}
and from (\ref{fit}) and Proposition \ref{prop:idc} also
\begin{align}
\lim_{u_{fin} \rightarrow \infty} \left({\Olin}\right)_{\Sf^\prime} \left(u_{fin},v,\theta,\phi\right)  = \lim_{u_{fin} \rightarrow \infty} \left({\Olin}\right)_{\Sf^\prime}  \left(u_{fin},v_0,\theta,\phi\right) = 0 \, 
\end{align}
and the angular commuted version
\begin{align}
\lim_{u_{fin} \rightarrow \infty} \slashed{\Delta} \left({\Olin}\right)_{\Sf^\prime} \left(u_{fin},v,\theta,\phi\right)  = \lim_{u_{fin} \rightarrow \infty} \slashed{\Delta} \left({\Olin}\right)_{\Sf^\prime} \left(u_{fin},v_0,\theta,\phi\right) = 0 \, .
\end{align}
\end{remark}

\begin{remark}
One can define the gauge function $f$ in (\ref{choicef}) also on the horizon $u_{fin}=\infty$ by taking an appropriate limit. This would recover the horizon-normalised gauge of \cite{DHR}. In this paper, however, we are only going to use $f$ defined for $u_{fin} <\infty$ and take the limit as $v\rightarrow \infty$. See Theorem \ref{theo:mtheo} below.
\end{remark}

\section{The main theorem} \label{sec:maintheo}
We are now ready to state the main theorem. We first define, for any $u>u_0$ the following initial data energy on $C_{u_0} \cup C_{v_0}$ associated with $\Si^\prime$ a partially initial data normalised solution supported on $\ell \geq 2$ which extends to null infinity:
\begin{align} \label{dataedef}
\mathcal{E}_{data} \left[\Si^\prime\right] \left(u\right) := & F_{u_0} \left[\Gamma, \Si^\prime\right] \left(v_0,\infty\right) + F_{v_0} \left[\Gamma, \Si^\prime\right] \left(u_0,u\right) \nonumber \\
&+\frac{1}{2} \lim_{v\rightarrow \infty} \int_{S^2} \sin \theta d\theta d\phi \left(r^3 \otxb \otx \left(u_0,v,\theta, \phi \right)\right) \, .
\end{align}
This energy is continuous in $u$, uniformly bounded for all $u>u_0$ (by the regularity of the solution near the horizon) and it can be computed explicitly from the data. We also define 
\begin{align} \label{dataedefl}
\mathcal{E}_{data} \left[\Si^\prime\right] :=\lim_{u \rightarrow \infty} \mathcal{E}_{data} \left[\Si^\prime\right] \left(u\right) \, .
\end{align}
Note that this limit is again well-defined by the regularity of the solution. Note also that at this point, we do not know whether $\mathcal{E}_{data} \left[\Si^\prime\right] \left(u\right)$ or $\mathcal{E}_{data} \left[\Si^\prime\right] $ is non-negative, however we will deduce the non-negativity of the total initial energy $\mathcal{E}_{data} \left[\Si^\prime\right]$ \emph{a posteriori} from the following theorem:

\begin{theorem} \label{theo:mtheo}
Consider $\Si^\prime$ a partially initial data normalised solution of the system of gravitational perturbations supported on $\ell \geq 2$ which is extendible to null infinity. Let the associated initial energy $\mathcal{E}_{data}\left[\Si^\prime\right]$ be as defined in (\ref{dataedef}). Then for any fixed $\infty>u_{fin} > u_0$ the following estimate holds:
\begin{align}
 & \ \ \ \ \ \ \  \ \ \int_{v_0}^{\infty} dv  \int_{S^2} r^2 \sin \theta d\theta d\phi \  |\xlin \Omega|^2  \left(u_{fin},v,\theta,\phi\right) \nonumber \\
+&\limsup_{v_{fin} \rightarrow \infty} \int_{u_0}^{u_{fin}} d u  \int_{S^2} r^2 \sin \theta d\theta d\phi   |\xblin |^2 \left(u,v_{fin},\theta,\phi\right)
 \leq \mathcal{E}_{data} \left[\Si^\prime\right] \left(u_{fin}\right)+ \mathcal{R}\left(u_{fin},v_0\right)  , \nonumber
\end{align}
with the remainder term on the right hand side defined in terms of the initial data as
\begin{align}
\mathcal{R}\left(u_{fin},v_0\right) &= \frac{1}{M} \int_{S^2} \sin \theta d\theta d\phi \Bigg[  4 |r^3 \left(\slashed{div}{\elin}_{} + \rlin_{}\right)|^2 
+ 16\Big| r \Delta_{S^2} \left( \Olin -\frac{r}{4\Omega^2} \otx\right)\Big|^2 \Bigg] \left(u_{fin},v_{0},\theta,\phi\right) \nonumber \\
&+\int_{S^2} \sin \theta d\theta d\phi \frac{1}{2} \Omega^{-2} r^3  \otx_{} \otxb_{}\left(u_{fin},v_{0},\theta,\phi\right)
\end{align}
and satisfying 
\[
\lim_{u_{fin} \rightarrow \infty} \mathcal{R}\left(u_{fin},v_0\right) =0 \, .
\]
\end{theorem}

\begin{proof}
As the proof will involve subtracting a pure gauge solution from the solution $\Si^\prime$, we will use subscripts $\Si^\prime$ to denote the geometric quantities associated with the solution $\Si^\prime$ for the duration of the proof, i.e.~we write $\left(\xlin\right)_{\Si^\prime}=\xlin$, $\otx_{\Si^\prime}=\otx$ etc.

Applying Proposition \ref{prop:conslaw} in the initial data gauge yields for any $u_{fin}, v_{fin}$ fixed the identity
\begin{align} \label{mit}
F_{u_{0}} \left[\Gamma, \Si^\prime \right] \left(v_0,v_{fin}\right) 
+ F_{v_0}\left[\Gamma, \Si^\prime \right] \left(u_0,u_{fin}\right)
=  F_{u_{fin}} \left[\Gamma, \Si^\prime \right] \left(v_0,v_{fin}\right) 
+ F_{v_{fin}}\left[\Gamma, \Si^\prime \right] \left(u_0,u_{fin}\right) \, .
\end{align}
The strategy now, roughly, is to establish positivity up to a boundary term of the terms on the right hand side and then to take that boundary term to the left. The resulting expression on the left will be converted to $\mathcal{E}_{data} \left[\Si^\prime\right] \left(u_{fin}\right)$ after taking the limit $v_{fin} \rightarrow \infty$. 

The details are as follows. Define
\[
\Sf^\prime = \Si^\prime - \mathscr{G} \, ,
\]
where $\mathscr{G}$ is the pure gauge solution generated by the $f$ chosen in (\ref{choicef}), cf.~Lemma \ref{lem:exactsol}. Using Proposition \ref{prop:gaugechange} and Lemma \ref{lem:newval}, the flux $F_{u_{fin}}\left[\Gamma, \Si^\prime \right] \left(v_0,v_{fin}\right)$ transforms according to (recall $f\left(v_0,\theta,\phi\right)=0$):
\begin{align} \label{sc1}
F_{u_{fin}}\left[\Gamma, \Si^\prime \right] &\left(v_0,v_{fin}\right) =  F_{u_{fin}}\left[\Gamma, \Sf^\prime\right] \left(v_0,v_{fin}\right) + \int_{S^2} \sin \theta d\theta d\phi \Bigg[ \left(\Omega^2 f\right)^2 \frac{6M}{r^2} \left(u_{fin},v_{fin}\right)
\nonumber \\ 
&- \frac{1}{2} \Omega^{-2} r^3 \left(\otx_{\Si^\prime}\otxb_{\Si^\prime}\right) \left(u_{fin},v_{fin}\right) +  \frac{1}{2} \Omega^{-2} r^3 \left(\otx_{\Si^\prime}\otxb_{\Si^\prime} \right) \left(u_{fin},v_{0}\right) \nonumber \\
& + \frac{f \Omega^2}{r} \left(2r^3 \left(+\slashed{div} \, {\elin} + {\rlin}\right)_{\Sf^\prime}\right) \left(u_{fin},v_{fin}\right)- \frac{f\Omega^2}{r} 2r^3 \left(\slashed{div} \, {\eblin} + \slashed{div} {\elin} \right)_{\Sf^\prime}\left(u_{fin},v_{fin}\right) \Bigg] \, , 
\end{align}
while Proposition \ref{prop:nullit} yields
\begin{align} \label{sc2}
F_{v_{fin}} \left[\Gamma, \Si^\prime \right] \left(u_0,u_{fin}\right) = & \int_{u_0}^{u_{fin}} d u  \int_{S^2} d\theta d\phi \ r^2 \sin \theta   \Big|\left(\xblin\right)_{\Si^\prime}\Big|^2 \left(u,v_{fin}\right)
\nonumber \\
&+\frac{1}{2}\int \sin \theta d\theta d\phi r^3  \otxb_{\Si^\prime} \otx_{\Si^\prime} \left(u_{fin},v_{fin}\right)
\nonumber \\
&-\frac{1}{2}\int \sin \theta d\theta d\phi r^3   \otxb_{\Si^\prime} \otx_{\Si^\prime} \left(u_{0},v_{fin}\right)  \nonumber \\
&+\textrm{terms vanishing in the limit $v_{fin} \rightarrow \infty$} \, .
\end{align}
Using Lemma \ref{lem:newval} and noting $\left({\xlin}\right)_{\Si^\prime}=\left(\xlin\right)_{\Sf^\prime}$ we simplify the sum of (\ref{sc1}) and (\ref{sc2}) to
\begin{align} \label{lemu}
F_{u_{fin}}\left[\Gamma, \Si^\prime \right] \left(v_0,v_{fin}\right) &+ F_{v_{fin}} \left[\Gamma, \Si^\prime \right] \left(u_0,u_{fin}\right)  = \nonumber \\
& \int_{v_0}^{v_{fin}} dv \int_{S^2} d\theta d\phi r^2 \sin \theta \left[ \Big|\left(\xlin\right)_{\Si^\prime} \Omega\Big|^2 + 2\Omega^2 \Big|\left({\eblin}\right)_{\Sf^\prime}\Big|^2 \right] \left(u_{fin},v\right) \nonumber \\ 
&+\int_{u_0}^{u_{fin}} d u  \int_{S^2} d\theta d\phi \ r^2 \sin \theta   \Big|\left(\xblin\right)_{\Si^\prime}\Big|^2 \left(u,v_{fin}\right) + Q \, ,
 \end{align}
 where the extra term $Q$ is given by (use Lemma \ref{lem:newval})
 \begin{align}
Q=  \int_{S^2} \sin \theta d\theta d\phi \Bigg[ &\left(\Omega^2 f\right)^2 \frac{6M}{r^2} \left(u_{fin},v_{fin}\right)
\nonumber \\ 
+ & \frac{1}{2} \Omega^{-2} r^3  \otx_{\Si^\prime} \otxb_{\Si^\prime} \left(u_{fin},v_{0}\right) 
- \frac{1}{2} r^3  \otxb_{\Si^\prime} \otx_{\Si^\prime}  \left(u_{0},v_{fin}\right)  \nonumber \\
&+ \frac{f \Omega^2}{r} \left(u_{fin},v_{fin}\right) \cdot \left(2r^3 \left(+\slashed{div}{\elin}_{} + \rlin_{}\right)_{\Si^\prime} -4 r \Delta_{S^2} \left({\Olin}\right)_{\Sf^\prime} \right) \left(u_{fin},v_{0}\right) \Bigg]\nonumber \\
&+\textrm{terms vanishing in the limit $v_{fin} \rightarrow \infty$} \, .
\end{align}
Note that a cancellation (up to a term vanishing in the limit $v_{fin} \rightarrow \infty$) has appeared between the first term in the second line of (\ref{sc1}) and the term in the second line of (\ref{sc2}). Applying the Cauchy-Schwarz inequality to the expression for $Q$ exploiting the positive first term we can estimate
\[
Q \geq -\mathcal{R} \left(u_{fin},v_0\right) -\frac{1}{2}\int_{S^2} \sin \theta d\theta d\phi r^3 \otx_{\Si^\prime} \otxb_{\Si^\prime} \left(u_{0},v_{fin}\right) +\textrm{terms vanishing as $v_{fin} \rightarrow \infty$} \, .
\]
Inserting this back into (\ref{lemu}) and combining (\ref{lemu}) with (\ref{mit}) we conclude after taking the limit $v_{fin} \rightarrow \infty$
\begin{align}
\mathcal{E}_{data} \left[\Si^\prime\right] \left(u_{fin}\right) + \mathcal{R} \left(u_{fin},v_0\right) &\geq 
 \int_{v_0}^{\infty} dv  d\theta d\phi r^2 \sin \theta \left[ \Big|\left(\xlin\right)_{\Si^\prime} \Omega\Big|^2 + 2\Omega^2 \Big|\left({\eblin}\right)_{\Sf^\prime} \Big|^2 \right] \left(u_{fin},v\right) \nonumber \\
&\ \ \ +\limsup_{v_{fin} \rightarrow \infty} \int_{u_0}^{u_{fin}} d u  \int_{S^2} d\theta d\phi \ r^2 \sin \theta   \Big|\left(\xblin\right)_{\Si^\prime}\Big|^2 \left(u,v_{fin}\right) \, ,
\end{align}
which proves the estimate claimed in the theorem. The conclusion about the limit of $\mathcal{R}\left(u_{fin},v_0\right)$ follows directly from the horizon gauge condition satisfied by $\Si^\prime$ and Remark \ref{rem:hozlimit} in conjunction with (\ref{fit}).
\end{proof}

\begin{remark}
The estimates derived in the proof of Theorem \ref{theo:mtheo} give more control than explicitly stated. In particular, the gauge function $f$ is controlled (we dropped a good term in the expression for $Q$) and so is $\left({\eblin}\right)_{\Sf^\prime}$. This will be exploited in future work.
\end{remark}

\begin{remark}
Taking the limit $u_{fin} \rightarrow \infty$ in Theorem \ref{theo:mtheo} we see that we must have $\mathcal{E}_{data}\left[\Si^\prime\right]\geq 0$. Note also that the formula for $\mathcal{E}_{data}\left[\Si^\prime\right]$ simplifies considerably for a fully initial data normalised solution.
\end{remark}

Taking the limit $u_{fin} \rightarrow \infty$ and using that the pointwise limit of the quantity $r\xblin$ actually exists on null infinity (in view of the solution $\Si^\prime$ being extendible to null infinity), we conclude control on the energy fluxes through the event horizon and null infinity:
\begin{corollary} \label{cor:mtheo}
With the assumptions of Theorem \ref{theo:mtheo} we have that
\begin{itemize}
\item  the total flux of the linearised shear on null infinity is bounded:
\begin{align} \label{xz}
 \int_{u_0}^{\infty} d u  \int_{S^2} d\theta d\phi \ r^2 \sin \theta   |\xblin|^2 \left(u,\infty,\theta,\phi\right) \leq \mathcal{E}_{data} \left[\Si^\prime\right] \, ,
\end{align}
\item the total flux of the linearised shear on the horizon is bounded:
\begin{align}
 \int_{v_0}^{\infty} dv  \int_{S^2}d\theta d\phi r^2 \sin \theta \left[ |\xlin\Omega|^2 \right] \left(\infty,v,\theta,\phi\right) \leq \mathcal{E}_{data} \left[\Si^\prime\right] \, .
\end{align}
\end{itemize}
\end{corollary}
Note that the quantity appearing on the left hand side of (\ref{xz}) is the total amount of gravitational radiation measured by far away observers.

%
%
%
%
%
%
%
%
%

\section{A second conservation law} \label{sec:2con}
We end the paper by stating a second conservation law. Unlike the first, it will involve curvature components, which is why we denote the corresponding fluxes by $F\left[\Gamma,R,\mathscr{S} \right]$. More precisely, we define
\begin{align} \label{vflux2}
F_v \left[\Gamma,R,\mathscr{S} \right] \left(u_1,u_2\right)= \int_{u_1}^{u_2} d u \int_{S^2} d\theta d\phi \sin \theta \Big[ 3Mr \olinb \otx -3M \left(1-\frac{4M}{r}\right)   \Olin \otxb \nonumber \\
+\frac{1}{2} \Omega^2 r^4 \left(|\rlin|^2 + |\slin|^2\right) -3Mr \Omega^2 |\elin|^2 + \frac{1}{2} r^4 \Omega^2 |\bblin|^2 \Big]
\end{align}
and
\begin{align} \label{uflux2}
F_u \left[\Gamma,R, \mathscr{S} \right] \left(v_0,v\right)= \int_{v_0}^v d v  \int_{S^2} d\theta d\phi \sin \theta \Big[ 3Mr \olin \otxb +3M \left(1-\frac{4M}{r}\right)   \otx \Olin \nonumber \\
+\frac{1}{2} \Omega^2 r^4 \left(|\rlin|^2 + |\slin|^2\right) -3Mr \Omega^2 |\eblin|^2 + \frac{1}{2} r^4 \Omega^2 |\blin|^2 \Big] \, .
\end{align}
The following conservation law holds:
\begin{proposition} \label{prop:conslaw2}
For any $u_0<u_1< u_2<\infty$ and $v_0 < v_1<v_2 < \infty$ we have the conservation law
\begin{align}
F_v\left[\Gamma,R, \mathscr{S} \right] \left(u_0,u_1\right)+ F_u \left[\Gamma,R,\mathscr{S}\right] \left(v_0,v_1\right) = F_{v_0} \left[\Gamma,R,\mathscr{S}\right] \left(u_0,u_1\right) +  F_{u_0}\left[\Gamma,R,\mathscr{S}\right] \left(v_0,v_1\right) \, .
\end{align}
\end{proposition}
\begin{proof}
Straightforward computation.
\end{proof}

We have the following analogue of Proposition \ref{prop:gaugechange}:
\begin{proposition} \label{prop:gaugechange2}
Let $f\left(v,\theta,\phi\right)$ be a smooth gauge function generating a pure gauge solution of the system of gravitational perturbations as in Lemma \ref{lem:exactsol}. Then the flux on fixed constant-$u$ hypersurfaces satisfies
\begin{align}
F_u \left[\Gamma_{}, R_{},\mathscr{S}\right] \left(v_0,v\right)= F_u \left[\Gamma,R,\tilde{\mathscr{S}}\right] \left(v_0,v\right) + \int_{S^2} \sin \theta d\theta d\phi\left( \mathcal{G} \left(v, u,\theta,\phi\right)-\mathcal{G} \left(v_0, u,\theta,\phi\right)\right) \nonumber
\end{align}
where
\begin{align}
\mathcal{G} = 3M r \left( \left(\Olin\right)_{\So} \ \otxb_{\So} - \left({\Olin}\right)_{\Sop} \ {\otxb}_{\Sop} \right)
- 3M^2  \left( \frac{f}{r} {\otx}_{\Sop} \right) \nonumber \\
+ 12M^2 \left(\frac{1}{r^2} f\Omega^2 \left({\Olin}\right)_{\Sop} \right) 
-3M f \Omega^2 r \left(\slashed{div} \, {\elin} + {\rlin} \right)_{\Sop} + \frac{3M}{2} \Big| \nabla_{S^2} \frac{f \Omega^2}{r} \Big|^2- \frac{6M^2}{r^3}  \left(f\Omega^2\right)^2   \, . \nonumber
\end{align} 
In other words, the differences of the fluxes in the old and in the new gauge is a pure boundary term.
\end{proposition}

\begin{proof}
Straightforward computation.
\end{proof}

Similar arguments to those presented in Sections \ref{sec:gaugechoice} and \ref{sec:maintheo} lead to control of additional fluxes. These will be exploited elsewhere.

\section{Acknowledgements}
I acknowledge support through an ERC grant. I am grateful to Martin Taylor and Thomas Johnson for checking many of the computations and to Mihalis Dafermos for several insightful comments and useful suggestions during the preparation of the manuscript.

\bibliographystyle{hacm}
\bibliography{conslawbib}
\end{document}